\documentclass[runningheads]{llncs}
\usepackage[T1]{fontenc}
\usepackage{graphicx}
\usepackage{hyperref}
\usepackage{color}

\usepackage{comment}

\usepackage{amsfonts}
\usepackage{mathtools}
\usepackage[style=alphabetic,natbib=true,maxcitenames=3]{biblatex}

\newcommand{\e}{\varepsilon}

\newcommand{\N}{\mathbb{N}}
\newcommand{\Z}{\mathbb{Z}}
\newcommand{\Q}{\mathbb{Q}}
\newcommand{\R}{\mathbb{R}}
\newcommand{\C}{\mathbb{C}}
\newcommand{\logn}{\log n}
\newcommand{\loglogn}{\log\logn}
\DeclareMathOperator{\supp}{supp}
\DeclareMathOperator*{\argmax}{arg\,max}

\DeclareMathOperator*{\E}{\mathbb{E}}
\DeclareMathOperator{\sign}{sign}

\DeclarePairedDelimiter{\iop}{(}{)}

\DeclarePairedDelimiter{\ios}{\{}{\}}
\DeclarePairedDelimiter{\absolute}{|}{|}
\DeclarePairedDelimiter{\normm}{\lVert}{\rVert}
\DeclarePairedDelimiter{\inprod}{\langle}{\rangle}
\newcommand{\norm}{\normm*}
\newcommand{\p}{\iop*}

\newcommand{\s}{\ios*}
\newcommand{\abs}{\absolute*}

\newcommand{\inp}{\inprod*}

\newcommand{\sym}{\text{Sym}}
\DeclareMathOperator{\NE}{NE}
\DeclareMathOperator{\PNE}{PNE}

\DeclareMathOperator{\br}{BR}

\spnewtheorem*{theorem*}{Theorem}{\normalshape\bfseries}{\itshape}

\begin{document}
\title{Randomness Requirements and Asymmetries in Nash Equilibria}
%
%
\author{Edan Orzech\and Martin Rinard}
%
\authorrunning{E. Orzech and M. Rinard}
%
\institute{MIT, Cambridge MA 02139, USA\\
\email{\{iorzech,rinard\}@csail.mit.edu}}
%
\maketitle              
\begin{abstract}
In general, Nash equilibria in normal-form games may require players to play (probabilistically) mixed strategies. We define a measure of the complexity of finite probability distributions and study the complexity required to play Nash equilibria in finite two player $n\times n$ games with rational payoffs.  
Our central results show that there exist games in which there is an exponential vs. linear gap in the complexity of the mixed distributions that the two players play in the (unique) Nash equilibrium of these games. This gap induces asymmetries in the amounts of space required by the players to represent and sample from the corresponding distributions using known state-of-the-art sampling algorithms. We also establish exponential upper and lower bounds on the complexity of Nash equilibria in normal-form games. These results highlight (i) the nontriviality of the assumption that players can play any mixed strategy and (ii) the disparity in resources that players may require to play Nash equilibria in normal-form games.

\end{abstract}
\section{Introduction}
A well-known result in game theory is the existence of Nash equilibria (NEs) in any finite normal-form game~\cite{nash1951non}. In general, NEs may require players to play (probabilistically) mixed strategies. This is required for reducing the existence of NEs to the existence of fixed points of a continuous function (where the continuum is the product set of mixed strategies of all players).
To establish the existence of a NE, players must be able to sample from the probability distribution corresponding to their strategy. 

We define a measure of complexity of finite probability distributions and study the 
associated randomness requirements and asymmetries in NEs in two-player $n\times n$ normal-form games
with rational payoffs. For a given finite distribution, this measure quantifies the amount of space required to represent the coefficients of the distribution assuming,
as is the case for known state-of-the-art sampling algorithms~\cite{knuth1976complexity,Devr86,walker1977efficient,vose1991linear,smith2002sample,bringmann2017efficient,saad2020fast,bringmann2013succinct}, that a distribution $\p{\frac{p_1}{q},...,\frac{p_n}{q}}$, where $p_1,...,p_n,q\in\N$, is stored as the binary representations of $p_1,...,p_n$ (the denominator is then the sum of the numerators $p_i$). It is important to understand this space requirement because, to sample from their NE mixed strategy distributions (and therefore play their corresponding mixed NEs), players using these sampling algorithms must store the coefficients of these distributions.

Our central results show that there exist games in which there is an exponential vs. linear gap in the complexity of the mixed distributions that the two players play at (in these games the unique) NE. This gap induces corresponding asymmetries in the amount of space required by the players to represent and sample from the corresponding distributions using known state-of-the-art sampling algorithms~\cite{knuth1976complexity,Devr86,walker1977efficient,vose1991linear,smith2002sample,bringmann2017efficient,saad2020fast,bringmann2013succinct}. We also establish exponential upper and lower bounds on the complexity of any NE in the games that we study. These results highlight (i) the nontriviality of the assumption that players can play any mixed strategy and (ii) the disparities in resources that players may require to play NEs in the games that we study.
We now discuss the motivation to our model, lay out the model, and then discuss our results and proof methods in greater detail.

\subsection{Model and motivation}\label{subsec:model}
To play a mixed strategy NE, players use exact sampling algorithms to implement their mixed strategies (and therefore require space sufficient to run these sampling algorithms). We work with finite probability distributions with rational coefficients (we call them {\it rational} distributions\footnote{As stated in \citet{avis2010enumeration}, for 2-player games with rational payoffs, every NE can be described as a convex combination of {\it extreme NEs}, and those constitute NE strategy profiles where all the mixed strategies are rational (as the extreme NEs are vertices of rational polytopes). In particular, if the game is nondegenerate, then the extreme NEs are the only NEs of the game.}) and consider a complexity measure $C$ that is the least common multiple of the denominators of the coefficients of the probability distribution. Working with the standard binary encoding of integers, $C$ bounds the amount of space required to store the coefficients: a distribution $x$ over $[n]$ with $C(x)=q$ requires at least $\Omega(n+\log q)$ bits (when $x$ is unbalanced) and at most $O(n\log q)$ bits (when $x$ is balanced).\footnote{See Theorem \ref{thm:space reqs} for a formal statement.}

Because $C$ bounds the number of bits required to store the coefficients of finite rational distributions, it also bounds space requirements for known state-of-the-art and commonly used exact sampling algorithms~\cite{knuth1976complexity,walker1977efficient,smith2002sample,vose1991linear,saad2020fast,bringmann2017efficient,Devr86,bringmann2013succinct}, which store the numerators of the coefficients of these distributions. Therefore, our results have direct implications on the space complexity of these algorithms, in the context of playing NEs in $n\times n$ games.

We now define $C$ formally.
\begin{definition}
\label{def:complexity}
Let $x=(x_1,...,x_n)$ be a distribution where $x_1,...,x_n\in\Q$. Represent $x$ uniquely as $\p{\frac{p_1}{q},...,\frac{p_n}{q}}$ where $p_1,...,p_n,q\in\N$, $q>0$ and $\gcd(p_1,...,p_n)=1$.\footnote{Here $0\in\N$.} Define the {\it complexity} of $x$ to be $C(x)=q$.
\end{definition}
For instance, $C\p{\p{\frac{25}{100},\frac{51}{100},\frac{24}{100}}}=100$.

In this paper we only consider rational distributions.
However, for completeness, if $x$ is a distribution that does not have this form, define $C(x)=\infty$. Each player $i$ has a capability $c_i$ that characterizes the amount of storage the player has available to store the coefficients. 

We define the capabilities of the players.
\begin{definition}\label{def:capability}
For player $i$, we say that $i$ has capability $c_i$ if they can play only strategies from the set $\s{x\in\Delta[n]\mid C(x)\le c_i}$, where $\Delta[n]$ is the set of distributions over $[n]$.
\end{definition}

We consider games where the players have different capabilities. This motivates our definition of the {\it bounded-randomness game}, as follows:
\begin{definition}\label{def:game}
In a $(c_1,c_2)$-bounded-randomness game, the two players have (integral) randomness capabilities $c_1,c_2\ge1$. The set of all strategies (pure and mixed) available to player $i$ is given by $\s{x\in\Delta[n]\mid C(x)\le c_i}$.
\end{definition}

The game is defined by two payoff matrices $A,B\in\Z^{n\times n}$ (see Section \ref{sec:defs and lemmas} for further details on this choice of games). For the set of NEs of the game, $\NE(A,B)$, let $C_1(A,B)=\min_{(x,y)\in\NE}C(x)$ and $C_2(A,B)=\min_{(x,y)\in\NE}C(y)$. I.e., $C_i(A,B)$ is the minimal capability required from player $i$ to play any NE strategy. By the best-response condition (see \cite{roughgarden2010algorithmic}, Chapter 3), and because pure strategy distributions have complexity 1, it follows that having players with capabilities at least $C_i(A,B)$ is a necessary condition for the existence of a NE in the corresponding bounded-randomness game (see Theorem \ref{thm:no ne without caps}).

\subsection{Our results}\label{subsec:results}
We study upper and lower bounds on the minimum capabilities $C_i(A,B)$ required to play NEs. Our main results are as follows. First, we prove an upper bound on $C_i(A,B)$ 
as a function of $n$ and the payoff matrices $A,B$.

\begin{theorem}\label{thm:2.2}
For every $A,B\in\Z^{n\times n}$ there holds
\begin{align*}
    C_1(A,B)\le\p{\max_{i,j}\abs{B_{i,j}}}^n\cdot2^{O(n\logn)},~
    C_2(A,B)\le\p{\max_{i,j}\abs{A_{i,j}}}^n\cdot2^{O(n\logn)}
\end{align*}
In particular, if $A,B$ have entries which are $poly(n)$, then $C_i(A,B)=2^{O(n\logn)}$.
\end{theorem}

We then prove lower bounds via explicit constructions. Each construction is also slightly modified to produce lower bounds for constant-sum games. The first bound is a $2^{\Tilde{\Omega}(\sqrt n)}$ lower bound. Let $\Z_2=\s{0,1}$.

\begin{theorem}\label{thm:2.3}
There is a function $f(n)=\Tilde\Theta(\sqrt n)$ such that
\begin{enumerate}
    \item\label{item:thm 2.3.1} for every $n$ there are $A,B\in\Z_2^{n\times n}$ such that $C_1(A,B)=2^{f(n)}$ and $C_2(A,B)=n$;
    \item\label{item:thm 2.3.2} for every $n$ there are $A,B\in\Z_2^{n\times n}$ such that the game $(A,B)$ is a constant-sum game and $C_1(A,B)=C_2(A,B)=2^{f(n)}$.
\end{enumerate}
\end{theorem}

To prove these bounds, we construct block matrices in which the block sizes are prime numbers and the lower bound is the product of those sizes.
In each block on its own, the players need to fully mix their strategies so as to play an equilibrium (i.e., pick each pure strategy with probability $>0$). We show that this property is preserved in the large game with these blocks as `subgames', so there is a unique NE and it is fully mixed.
In addition, the proof uses some well-known properties of prime numbers. Afterwards, using a different idea, we strengthen the result with a $2^{\Omega(n)}$ lower bound:

\begin{theorem}\label{thm:2.4}
There is a function $g(n)=\Theta(n)$ such that
\begin{enumerate}
    \item\label{item:thm 2.4.1} for every $n$ there are $A,B\in\Z_2^{n\times n}$ such that $C_1(A,B)=2^{g(n)}$ and $C_2(A,B)=n$;
    \item\label{item:thm 2.4.2} for every $n$ there are $A,B\in\Z_2^{n\times n}$ such that the game $(A,B)$ is a constant-sum game and $C_1(A,B)=C_2(A,B)=2^{g(n)}$.
\end{enumerate}
\end{theorem}

To prove these bounds, we study a pair of linear recurrences. These recurrences have appealing properties which align with our goal of proving the lower bound. For example, the set of NEs of this construction is characterized by linear dependencies in these recurrences.
More specifically, we show that the players need to fully mix their equilibrium strategies when the particular linear dependencies do not exist, and these dependencies indeed do not appear when the game matrices are of size $n\times n$ for $n\ge8$.
The lower bound constructions in Theorem \ref{thm:2.3}, Item \ref{item:thm 2.3.1} and Theorem \ref{thm:2.4}, Item \ref{item:thm 2.4.1} are instances of what are sometimes called {\it imitation games} \cite{mclennan2005imitation,mclennan2010simple},
since player 1 gains a payoff of 1 iff its pure strategy matches that of player 2.
These results also reveal an exponential vs. linear asymmetry in the randomness requirements from the two players.

Our proofs for Theorem \ref{thm:2.3}, Item \ref{item:thm 2.3.2} and Theorem \ref{thm:2.4}, Item \ref{item:thm 2.4.2} require an additional step. There, we exploit the high symmetry in the lower bound constructions of the corresponding Item 1s. We also use the high symmetry in the set of NEs of zero-sum games.

\begin{example}\label{example:1}
Consider the following $8\times8$ game game defined by the following payoff matrices:
\begin{align}\label{eq:example1}
    A=\left(\begin{array}{cccccccccccccccc}
    1&&0&&0&&0&&0&&0&&0&&0\\
    0&&1&&0&&0&&0&&0&&0&&0\\
    0&&0&&1&&0&&0&&0&&0&&0\\
    0&&0&&0&&1&&0&&0&&0&&0\\
    0&&0&&0&&0&&1&&0&&0&&0\\
    0&&0&&0&&0&&0&&1&&0&&0\\
    0&&0&&0&&0&&0&&0&&1&&0\\
    0&&0&&0&&0&&0&&0&&0&&1
    \end{array}\right),~~~
    B=\left(\begin{array}{cccccccccccccccc}
    0&&1&&1&&0&&0&&0&&0&&0\\
    0&&0&&1&&1&&0&&0&&0&&0\\
    0&&1&&0&&1&&1&&0&&0&&0\\
    0&&0&&1&&0&&1&&1&&0&&0\\
    0&&0&&0&&1&&0&&1&&1&&0\\
    0&&0&&0&&0&&1&&0&&1&&1\\
    0&&0&&0&&0&&0&&1&&0&&1\\
    1&&0&&0&&0&&0&&0&&0&&0\end{array}\right)
\end{align}
Then the unique NE is given by the pair of distributions $(x,y)$ of players 1, 2 respectively, where
\begin{align*}
    x=\p{\frac{6}{34},\frac{2}{34},\frac{3}{34},\frac{1}{34},\frac{4}{34},\frac{5}{34},\frac{4}{34},\frac{9}{34}},~y=\p{\frac{1}{8},\frac{1}{8},\frac{1}{8},\frac{1}{8},\frac{1}{8},\frac{1}{8},\frac{1}{8},\frac{1}{8}}
\end{align*}
$C(x)=34,~C(y)=8$, therefore $C_1(A,B)=34,~C_2(A,B)=8$. 
\end{example}

The lower bound constructions in Theorem \ref{thm:2.3}, Item \ref{item:thm 2.3.2} and Theorem \ref{thm:2.4}, Item \ref{item:thm 2.4.2} are the corresponding constant-sum versions of the games that we use to prove the lower bounds in 
Theorem \ref{thm:2.3}, Item \ref{item:thm 2.3.1} and Theorem \ref{thm:2.4}, Item \ref{item:thm 2.4.1}.

\begin{example}\label{example:2}
Consider the $8\times8$ constant-sum game defined by the following payoff matrices: 
\begin{align}\label{eq:example2}
    A'=\left(\begin{array}{cccccccccccccccc}
    1&&0&&0&&1&&1&&1&&1&&1\\
    1&&1&&0&&0&&1&&1&&1&&1\\
    1&&0&&1&&0&&0&&1&&1&&1\\
    1&&1&&0&&1&&0&&0&&1&&1\\
    1&&1&&1&&0&&1&&0&&0&&1\\
    1&&1&&1&&1&&0&&1&&0&&0\\
    1&&1&&1&&1&&1&&0&&1&&0\\
    0&&1&&1&&1&&1&&1&&1&&1
    \end{array}\right),~~~
    B'=\left(\begin{array}{cccccccccccccccc}
    0&&1&&1&&0&&0&&0&&0&&0\\
    0&&0&&1&&1&&0&&0&&0&&0\\
    0&&1&&0&&1&&1&&0&&0&&0\\
    0&&0&&1&&0&&1&&1&&0&&0\\
    0&&0&&0&&1&&0&&1&&1&&0\\
    0&&0&&0&&0&&1&&0&&1&&1\\
    0&&0&&0&&0&&0&&1&&0&&1\\
    1&&0&&0&&0&&0&&0&&0&&0\end{array}\right)
\end{align}
the unique NE is given by the pair of distributions $(x',y')$ of players 1, 2 respectively, where
\begin{align*}
    x'=\p{\frac{6}{34},\frac{2}{34},\frac{3}{34},\frac{1}{34},\frac{4}{34},\frac{5}{34},\frac{4}{34},\frac{9}{34}},~y'=\p{\frac{9}{34},\frac{4}{34},\frac{5}{34},\frac{4}{34},\frac{1}{34},\frac{3}{34},\frac{2}{34},\frac{6}{34}}
\end{align*}
$C(x')=C(y')=34$, therefore $C_1(A',B')=C_2(A',B')=34$.
\end{example}

\paragraph{Relationship with entropy} The entropy is often used to bound players' randomness capabilities \cite{halpern2015algorithmic,neyman2000repeated} (see Section \ref{sec:papers} for more details). The entropy of a distribution tells us exactly (up to an additive constant) how many random bits are needed to be sampled, in order to simulate the distribution exactly \cite{knuth1976complexity,saad2020fast}. Compared to our measure, the entropy does not address the question of space complexity. While the entropy provides a lower bound on the amount of space required, optimal sampling algorithms (in the number of random bits needed) can require much more space than this lower bound~\cite{knuth1976complexity,saad2020fast}. Indeed, our measure of the complexity of a distribution is independent, in some sense, of the entropy of the distribution: for example, the distributions $\p{\frac{1}{2},\frac{1}{2}}$ and $\p{\frac{1}{2}+\delta,\frac{1}{2}-\delta}$ have entropies $1$ and $1-\e$ respectively (for some small $\e>0$), but the latter distribution requires arbitrarily large space as $\delta\to0$, with respect to $C$, while the former requires only constant space.

\paragraph{Relationship with approximate NEs}
In general, a distribution might need arbitrarily large space to represent it accurately, for example, the distribution $\p{\frac{1}{2}+\delta,\frac{1}{2}-\delta}$ 
(it is straightforward to construct a $2\times2$ game where this is the unique NE distribution), as $\delta\to0$.\footnote{Note that by Theorem \ref{thm:2.2}, the payoff matrix will contain arbitrarily large entries.} Therefore, one might try to circumvent this space requirement. However, as we discuss in the next paragraph below, this attempt to circumvent the space requirement incurs increased time and sampling complexities. As a result, the measure $C$ essentially quantifies an inherent requirement of resources from a player who wishes to sample from a given distribution. Therefore, studying $C$, and framing players' capabilities with it, could be beneficial to improve our understanding of the role of randomness in NEs.

One could attempt to circumvent this space requirement by approximate sampling to play an additive $\e$-approximate NE ($\e$-NE, using \cite{roughgarden2010algorithmic} for reference). Given a NE distribution, truncating digits from its binary representation yields an $\e$-NE distribution that may use much fewer bits than in the given exact NE. A player then needs to approximate an equilibrium mixed strategy without knowing the strategy, as knowledge implies storage. The problem then reduces to a search problem, where the goal is to find an $\e$-NE. However, this creates time-complexity issues: \citet{rubinstein2015inapproximability} showed that finding an $\e$-NE is PPAD-complete in general games, for any $\e$ smaller than a global constant $c$. In other words, there is no PTAS for solving the problem. In this paper we will present constructions of {\it imitation games} \cite{mclennan2005imitation} (see the examples below, and the further discussion on these games in Section \ref{sec:papers}). Yet even for imitation games, \citet{murhekar2020approximate} showed that finding an $\e$-NE is PPAD-complete, for any $\e=n^{-c}$, for any $c>0$. I.e., there is no FPTAS for solving the problem.

\subsection{Organization}
In Section \ref{sec:papers} we overview related work.
In Section \ref{sec:defs and lemmas} we define the problem setup.
In Section \ref{sec:lemmas} we prove basic lemmas (proofs in Appendix \ref{sec:proofs}).
In Section \ref{sec:binary payoffs} we prove the upper and lower bounds, and in particular prove Theorems \ref{thm:2.2}, \ref{thm:2.3} and \ref{thm:2.4} (proofs in Appendices \ref{appendix:proof upper bound} \ref{appendix:proofs 2^sqrt n} and \ref{appendix:proofs 2^n}).
In Appendix \ref{sec:special} consider the special cases of $2\times2$ games and permutation games.

\section{Related Work}\label{sec:papers}

\subsection{Sampling Algorithms}

Modern algorithms for sampling from discrete rational probability distributions take as input the numerators of the corresponding probabilities (the denominator is then the sum of the numerators)~\cite{knuth1976complexity,walker1977efficient,smith2002sample,vose1991linear,saad2020fast,bringmann2017efficient,Devr86,bringmann2013succinct}. 
The algorithms preprocess the input to construct an auxiliary data structure (such as a binary decision tree or a prefix sum array), then use this data structure to convert a sequence of random bits into a sample from the specified finite distribution. 

For these algorithms the input size (with a standard binary encoding of the integers) for the NE distribution of player 2, namely the uniform distribution, is $\Theta(n)$, while the input size for the NE distribution of player 1 in Theorem \ref{thm:2.3}, Item \ref{item:thm 2.3.1} is $\Tilde\Theta(n\sqrt n)$ (Corollary \ref{cor:x1 is balanced}) and in Theorem \ref{thm:2.4}, Item \ref{item:thm 2.4.1} it is $\Theta(n^2)$ (Corollary \ref{cor:x2 is balanced}). The sizes of the auxiliary data structures vary depending on the algorithm, specifically $\Theta(n\logn)$ for the uniform distribution and $\Tilde\Theta(n\sqrt n),\Theta(n^2)$~\cite{saad2020fast,smith2002sample,vose1991linear}; $r+\Theta(\logn)$ for the uniform distribution and $r+\Tilde\Theta(\sqrt n),r+\Theta(n)$ for the said distributions, where $r\le n$ is a space-time tradeoff parameter~\cite{bringmann2017efficient,bringmann2013succinct}.
These results highlight the differences in resources required to represent, transmit, and work with the different probability distributions we consider in this paper. 

\subsection{Playing Games with Bounded Entropy}
In the context of playing games with bounded randomness, the most common approach to quantify the complexity of mixed strategies is via their entropy. There are two lines of research: one-shot games and repeated games.

\paragraph{One-shot games}
\citet{halpern2015algorithmic} consider games with costly computation, where players also aim to minimize their cost (either absolute or relative to other players). The cost can be manifested in the amount of randomness (random bits) the players use in their strategies. They show that the existence of a NE depends on the players' abilities to use randomness, and that NEs always exist in games where randomization is free (yet other computational aspects can still be costly).

Many papers in this context focus on $\e$-NEs (namely additive $\e$-approximate NEs) with small supports \cite{barany2007nash,feder2007approximating,anbalagan2015large,daskalakis2009note,lipton2003playing,goldberg1968probability}. Such an approximate NE will also have a small entropy, but a small support is not necessary for having small entropy. It is known \cite{daskalakis2009note} that uniform strategies with support of size at most two are sufficient to obtain a $\frac{1}{2}$-approximate NE, and \citet{feder2007approximating} show that to achieve a better approximation factor one needs to consider mixed strategies with supports of size $\Omega(\logn)$. On the other hand, \citet{lipton2003playing} prove the existence of $\e$-NEs with supports of size $O\p{\frac{\logn}{\e^2}}$ (for any $\e>0$) by constructing a quasi-polynomial time algorithm that finds one.
\citet{anbalagan2015large} refute the existence of $\e$-well-supported NEs\footnote{Strategy profiles where each strategy is comprised of pure strategies each of which is an $\e$-best response to the other strategies in the profile.} with supports of size $k$ for any constant $k$.
There is also work on the existence of pure NEs in random games \cite{goldberg1968probability} and the existence of NEs with supports of size two \cite{barany2007nash}.

\paragraph{Repeated games with bounded entropy}
Repeated games with bounded entropy have been studied in the last 25 years \cite{budinich2011repeated,hubavcek2016can,kalyanaraman2007algorithms,neyman2000repeated,valizadeh2019playing,valizadeh2019playing1}. First, there is the work of \citet{neyman2000repeated}, where they characterize the maxmin value of a repeated zero-sum game, when a player who is restricted to mixed strategies with (sufficiently) bounded entropy plays against an unrestricted player. \citet{budinich2011repeated} analyze a repeated version of the matching pennies game, derive results on the amount of entropy needed from both players to play an $\e$-NE, and study their optimal payoffs when they have bounded entropies.
\citet{hubavcek2016can} generalize those results (in both aspects) to a large set of repeated games that contains 2-player repeated zero-sum games. In both papers, the conclusion is that these games require $\Omega(N)$ random bits to play a NE, where $N$ is the number of repetitions, while other games require only $O(1)$ random bits. \citet{arthaud2023playing} advances this line of work by establishing a 0-1 law: in every repeated game, there is either a NE that requires only $O(1)$ entropy from all players, or every NE of the repeated game requires $\Omega(n)$ entropy from at least one of the players.
\citet{kalyanaraman2007algorithms} give both randomized and derandomized algorithms to find sparse $\e$-NEs (where strategies have small supports) in a learning settings where players have limited access to random bits.
\citet{valizadeh2019playing,valizadeh2019playing1} study the maxmin value in 2-player zero-sum repeated games in settings where one player is restricted to playing strategies with bounded entropy, while the second player is unrestricted. They study both asymptotic and nonasymptotic aspects of the value.
While the entropy captures the amount of random bits needed to sample from a given distribution, our measure of complexity captures the amount of space required to store that distribution, as discussed in Section \ref{subsec:results}.

\subsection{Imitation Games}
Our constructions involve {\it imitation games}, which are games where one player receives positive payoff iff the player plays the same pure strategy as the other player (for some notion of `same'). These games were first defined in \cite{mclennan2005imitation}, where they were leveraged to derive a simple proof of Kakutani's fixed-point theorem \cite{kakutani1941generalization}. Computing a NE was shown to be as hard in imitation games as in symmetric games, and it is known \cite{gale1950symmetric} that computing a NE in symmetric games is as hard as in general normal-form games. In both reductions the NE profiles remain the same as distributions, up to adding zero-probability events or swapping one of the mixed strategies in the profile with a uniform distribution. In \cite{mclennan2010simple}, the authors continue utilizing imitation games to derive additional simple proofs in computational complexity.

We note that the $4\times4$ instance of our constructions first appeared in \cite{bonifaci2008complexity}. There, the authors established the NP-completeness of finding a NE in which the mixed strategies are uniform on their supports. There, the authors use the $4\times4$ instance to construct a larger game, to reduce the problem they study to 3-SAT. However, their paper is not concerned with the space requirement issues that we address.
Moreover, in the constructions of \citet{bonifaci2008complexity}, the space requirements are smaller as well as the gaps in the requirements between the two players in the game. In contrast, our constructions exhibit much larger space requirements and asymmetries in the requirements between the players in the NEs, which highlights how large these requirements and asymmetries can get in normal-form games.

\section{Problem Setup}
\label{sec:defs and lemmas}
\paragraph{2-player normal-form games:}
We consider 2-player $n\times n$ normal-form games with nonnegative integer payoffs (see \cite{roughgarden2010algorithmic} for reference). We write such a game as $(A,B)$, where $A,B\in\Z^{n\times n}$ are the payoff matrices of players 1, 2 respectively. This makes $[n]$ the set of pure strategies for the two players. These induce the payoff functions $u_1,u_2:[n]\times[n]\to\N$ defined by $u_1(i,j)=A_{i,j},u_2(i,j)=B_{i,j}$. Let $\Delta X$ be the set of distributions over $X$. Then $u_1,u_2$ are naturally extended to payoffs of mixed strategy profiles $u_1,u_2:\Delta[n]\times\Delta[n]\to\R$ by $u_1(\sigma_1,\sigma_2)=\E_{(i,j)\sim(\sigma_1,\sigma_2)}(u_1(i,j))=\sigma_1^TA\sigma_2$ and similarly $u_2(\sigma_1,\sigma_2)=\sigma_1^TB\sigma_2$. For a finite set $X$, let $U_X$ be the uniform distribution on $X$. In addition to the definition of $\NE(A,B)$, let $\PNE(A,B)$ be the set of pure NEs of the game. When $A,B$ are known from context, we simply write $\NE$, $\PNE$.

As for the choice of these games, we assume that the payoffs are rational, and because (expected) payoff functions are linear, we further assume wlog that the payoffs are in $\N$. We will also consider {\it constant-sum games}, i.e., $A,B$ that satisfy $A+B=u{\bf1}_{n\times n}$ where $u\in\N$ is constant. The argument above shows that this is wlog compared to {\it zero-sum games} in the setting of rational payoffs.

\paragraph{Binary-payoff games:}
To prove our lower bounds, we will focus for the most part on these games. We will further assume that $\PNE(A,B)=\emptyset$, for otherwise $C_i(A,B)=1$. This means that $A_{i,j}\cdot B_{i,j}=0$ and $\sum_jA_{i,j},\sum_iB_{i,j}>0$ for all $i,j$.

\paragraph{Linear algebra:}In this paper we consider only $n\times n$ payoff matrices, and sometimes treat them as matrices in $\Z^{n\times n}$.
Let $A\in\Z^{n\times n}$. For $I,J\subseteq[n]$, let $A_{I,J}=(A_{i,j})_{i\in I,j\in J}\in\Z^{|I|\times|J|}$. Denote with $A[i|j]:=A_{[n]\setminus\s{i},[n]\setminus\s{j}}$ the $(i,j)$th minor of $A$. Let $\kappa_{i,j}:=(-1)^{i,j}\det(A[i|j])$ be the $(i,j)$th cofactor of $A$. Note that $\det A=\sum_{i,j}A_{i,j}\kappa_{i,j}$ and $(A^{-1})_{i,j}=\frac{1}{\det A}\kappa_{j,i}$. Now we define $K(A):=\sum_{i,j}\kappa_{i,j}$. For $i\in[n]$ and $b\in\Z^n$, let $A\stackrel{i}{\gets}b$ result from $A$ by replacing column $i$ with $b$. Then we get $K(A)=\sum_i\det(A\stackrel{i}{\gets}\bf1)$. The quantity $K(A)$ can be used to upper bound $C_2(A,B)$, and similarly with $K(B),C_1(A,B)$ (see Theorem \ref{thm:K upper bound}). Denote with $\bf1,\bf0$ the all-ones, all-zeros vectors of length $n$. For $i\in[n]$, let $e_i\in\R^n$ be 1 at entry $i$ and 0 everywhere else. For $I\subseteq[n]$, let $=\sum_{i\in I}e_i$.

\section{Basic Properties}\label{sec:lemmas}
The proofs of these lemmas can be found in Appendix \ref{sec:proofs}. We begin with two lemmas about games, then prove linear algebra lemmas.
\begin{lemma}\label{lmm:I subseteq J}
Any $(x,y)\in\NE(I_n,B)$ satisfies $\supp x\subseteq\supp y$ and $y=U_{\supp y}$.
\end{lemma}

\begin{lemma}\label{lmm:game inflation}
Let $A,B\in\Z^{n\times n}$, such that $A$ does not contain a column of 0s, and $B$ does not contain a row of 0s. Then it is possible to add a dummy strategy to each player which does not change the set of NEs (up to projecting the NE mixed strategies on the first $n$ pure strategies). In particular, if the game is constant-sum with $A+B={\bf1}_{n\times n}$, then it will remain constant-sum after adding the dummy strategies.
\end{lemma}

\begin{lemma}\label{lmm:1}
Suppose that $A$ is binary, $\det A\ne0$ and $A^{-1}\bf1\ge0$. Then
$|\det A|\le|K(A)|\le n|\det A|$.
\end{lemma}

It is strengthened as follows.
\begin{lemma}\label{lmm:2}
\label{lmm:K det inequality strengthened}
Assume the above and also that for some $M\ge0$, $\sum_jA_{i,j}\le M$ for all $i\in[n]$. Then $|K(A)|\ge\frac{n}{M}|\det A|$.
\end{lemma}

To demonstrate how $K(A)$ can show up in our context, we have the following:
\begin{theorem}\label{thm:K upper bound}
Let $A,B\in\Z_2^{n\times n}$. Suppose that $\det A,\det B\ne0$ and that the game admits a fully-mixed NE. Then $C_1(A,B)\le K(B),C_2(A,B)\le K(A)$.
\end{theorem}

We end with two properties of direct relevance to our motivation discussion in the introduction. The first states that if the players do not both have sufficient capability to play a (specific) NE, then the induced game has no NEs.

\begin{theorem}\label{thm:no ne without caps}
Let $(A,B)$ be a game, and let $c_1,c_2\ge1$. The corresponding $(c_1,c_2)$-bounded-randomness game has a NE iff there is an $(x,y)\in\NE(A,B)$, such that $c_1\ge C(x)$ and $c_2\ge C(y)$.
\end{theorem}
\begin{proof}
By the best-response condition (\cite{roughgarden2010algorithmic}), for a given strategy profile we only need to consider deterministic, unilateral deviations. These are valid strategies in the $(c_1,c_2)$-bounded-randomness game, for every choice of parameters. Therefore, a profile $(a,b)$ with $C(a)\le c_1,C(b)\le c_2$ is a NE of the bounded-randomness game iff $(a,b)\in\NE(A,B)$. The statement now follows.
\qed\end{proof}

The second property establishes the correlation between $C$ and the Shannon entropy, as was discussed in the introduction. We now establish space complexity bounds on storing the coefficients of a rational distribution.

\begin{theorem}\label{thm:space reqs}
Let $\sigma=\p{\frac{p_1}{q},...,\frac{p_n}{q}}$ with $C(\sigma)=q$. Then storing $p_1,...,p_n$ in their binary representations requires between $\Omega(n+\log q)$ space and $O(n\log q)$ space.
\end{theorem}
\begin{proof}
Each $p_i$ requires at least 1 bit. Also, storing $p_1,...,p_n$ requires at least as much space as storing $\sum p_i=q$, which requires $\log q$ bits. Therefore we need $\Omega(n+\log q)$ space. On the other hand, since $p_1,...,p_n\le q$, each of them requires at most $\log q$ bits of space. Therefore we need $O(n\log q)$ space.
\qed\end{proof}

\section{Worst-Case Bounds}
\label{sec:binary payoffs}
We begin with stating Theorem \ref{thm:2.2} formally (see proof in Appendix \ref{appendix:proof upper bound}).
\begin{theorem}[Restatement of Theorem \ref{thm:2.2}]\label{thm:upper bound binary}
Let $A,B\in\Z^{n\times n}$. Then
\begin{align*}
    C_1(A,B)\le\p{\max_{i,j}\abs{B_{i,j}}}^n\cdot2^{O(n\logn)},~
    C_2(A,B)\le\p{\max_{i,j}\abs{A_{i,j}}}^n\cdot2^{O(n\logn)}
\end{align*}
In particular, if $A,B$ have entries which are $poly(n)$, then $C_i(A,B)=2^{O(n\logn)}$.
\end{theorem}

We now overview the lower bounds in Theorems \ref{thm:2.3} and \ref{thm:2.4}. To obtain the lower bounds, we note that in all constructions and for all $n\ge8$, there is a unique NE and it is fully mixed. We then derive closed-form expressions for the unique NE in each construction, and analyze its complexity in terms of $C$. We begin with the simpler lower bound of $2^{\Tilde\Omega(n)}$, and then proceed to the more involved $2^{\Omega(n)}$ lower bound.

It is worth mentioning that \citet{alon1997anti} constructed matrices $B_n$ for every $n$, where the game $(I_n,B_n)$ has a unique fully-mixed NE $(x_n,U_n)$ which satisfies $C(x_n)=2^{\Theta(n\logn)}$. However, these games also admit many pure NEs, implying that $C_i(I_n,B_n)=1$. Put formally,

\begin{theorem*}[{\citep[Theorem 3.5.1]{alon1997anti}}]
For every $n$ there is a matrix $B_n\in\{0,1\}^{n\times n}$ such that for the game $(I_n,B_n)$ there exists a fully-mixed NE $(x,y)$ where $C(x)=2^{\frac{1}{2}n\logn-n(2+o(1))}$, while $C(y)=n$.
\end{theorem*}

\subsection{$2^{\Tilde\Omega(\sqrt n)}$ Lower Bound}\label{subsec:2^sqrt n}
Proofs for this subsection can be found in Appendix \ref{appendix:proofs 2^sqrt n}.
We now construct some matrices. We will then use these matrices and their properties to prove Theorem \ref{thm:2.3}, Item \ref{item:thm 2.3.1}. For $k\in\N$, let $B^k\in\Z_2^{k+1\times k+1}$ be the matrix where $B^k_{ij}=1$ if $i\ne j+1\bmod k+1$ and 0 otherwise. For example,
\begin{align*}
    B^5=\left(\begin{array}{cccccccccccc}1&&1&&1&&1&&1&&0\\0&&1&&1&&1&&1&&1\\1&&0&&1&&1&&1&&1\\1&&1&&0&&1&&1&&1\\1&&1&&1&&0&&1&&1\\1&&1&&1&&1&&0&&1\end{array}\right)
\end{align*}

We now present the main technical theorem of this subsection. This theorem establishes Theorem \ref{thm:2.3}, Item \ref{item:thm 2.3.2} for infinitely many $n$s. 

\begin{theorem}\label{thm:2^sqrt n lower bound}
Let $n\in\N$. For $N=\Theta(n^2\logn)$ there exist matrices $A,B\in\Z_2^{N\times N}$ such that $A=I_N$, and
\begin{align}
    C_1(A,B)=2^{\Theta(n\logn)}=2^{\Tilde{\Theta}(\sqrt N)},~C_2(A,B)=N
\end{align}
\end{theorem}

It is proved by creating a block matrix that has the blocks $B^{p_1},...,B^{p_n}$ for $p_1,...,p_n$ the first $n$ primes, and placing them diagonally on the block matrix albeit slightly shifted. This shift helps make the game `asymmetric' so that the only equilibrium is fully mixed.

\begin{corollary}[Restatement of Theorem \ref{thm:2.3}, Item \ref{item:thm 2.3.1}]\label{cor:sqrt n lower bound}
There is a function $f(n)=\Tilde\Theta(\sqrt n)$ such that for every $n$ there are $A,B\in\Z_2^{n\times n}$ such that $C_1(A,B)=2^{f(n)}$ and $C_2(A,B)=n$.
\end{corollary}
\begin{proof}
Follows by taking the games above and padding them with dummy strategies as in Lemma \ref{lmm:game inflation}.
\qed\end{proof}

The final corollary of this subsection states that the equilibrium distributions in Theorem \ref{thm:2^sqrt n lower bound} require the maximal amount of space (asymptotically) to store them.

\begin{corollary}\label{cor:x1 is balanced}
For every $n$, in the $n\times n$ game from Corollary \ref{cor:sqrt n lower bound}, the NE distribution $x$ of player 1 satisfies $\max_{i,j}\frac{x_i}{x_j}=\Tilde\Theta(\sqrt n)$. Therefore representing the numerators of $x$ would require $\Tilde\Theta(n\sqrt n)$ bits.
\end{corollary}

\subsection{$2^{\Omega(n)}$ Lower Bound}\label{subsec:2^n}
In this subsection we prove Theorem \ref{thm:2.4}, Item \ref{item:thm 2.4.1}. Proofs can be found in Appendix \ref{appendix:proofs 2^n}.
We first establish constructions and properties. For $z\in\C$, let $\Re(z),\Im(z)$ be the real part and imaginary part of $z$. We begin with a pair of linear recurrences.

\begin{definition}
Define the linear recurrences $a_n,b_n$ by
\begin{align*}
    a_1=1,a_2=1,a_3=1,a_4=0\\
    b_1=0,b_2=1,b_3=0,b_4=1
\end{align*}
and for $n\ge5$
\begin{align*}
    a_n=a_{n-2}-a_{n-3}+a_{n-4},\\
    b_n=b_{n-2}-b_{n-3}+b_{n-4}
\end{align*}
\end{definition}
In Appendix \ref{appendix:proofs 2^n} we prove some basic properties of these recurrences. Among them is the following property that connects the two recurrences:

\begin{lemma}\label{lmm:anbn}
For every $n$, $b_n+b_{n+1}=a_n$.
\end{lemma}

This property will be used later on to lower-bound the complexity of the equilibrium strategies we find.

The characteristic polynomial of the recurrence $b_n$ is $x^4-x^2+x-1=(x-1)(x^3+x^2+1)$. Let 
\begin{align}
    \rho&=\frac{1}{3}\p{-1-\sqrt[3]{\frac{2}{29-3\sqrt{93}}}-\sqrt[3]{\frac{1}{2}\p{29-3\sqrt{93}}}}\approx-1.4656,\\\nonumber
    z&=\frac{1}{12}\p{-4+\p{1+i\sqrt3}\sqrt[3]{4\p{29-3\sqrt{93}}}+\p{1-i\sqrt3}\sqrt[3]{4\p{29+3\sqrt{93}}}}\\
    &\approx0.2328-0.7926i
\end{align}
be the real root and one of the two complex roots of the polynomial $x^3+x^2+1$. The remaining root is $\overline{z}$.
We have $|z|\approx0.826$.
Along with the additional root of 1, we get the set of linear equations in $w_0,w_1,w_2,w_3$:
\begin{align}
    \begin{cases}
        w_0+w_1\rho+w_2z+w_3\overline{z}=0=b_1\\
        w_0+w_1\rho^2+w_2z^2+w_3\overline{z}^2=1=b_2\\
        w_0+w_1\rho^3+w_2z^3+w_3\overline{z}^3=0=b_3\\
        w_0+w_1\rho^4+w_2z^4+w_3\overline{z}^4=1=b_4\\
    \end{cases}
\end{align}
From these equations one finds
\begin{align}
    w_0=\frac{1}{3},~w_1\approx0.169,~w_2\approx-0.251+0.02i,~w_3=\overline{w_2}
\end{align}
and this gives the closed-form expression for $b_n$:
\begin{align}
\label{eq:bn}
    b_n=\frac{1}{3}+w_1\rho^n+w_2z^n+\overline{w_2}\cdot\overline{z}^n=\frac{1}{3}+w_1\rho^n+2\Re(w_2z^n)
\end{align}

We now have the two main properties of the recurrences $a_n,b_n$ which will allow us to characterize the NEs of the game that will be constructed:

\begin{lemma}
\label{prop:ratio sequences to rho}
There exists a strictly increasing sequence $\lambda_n<\rho$ and a strictly decreasing sequence $\mu_n>\rho$ such that for every $n\ge5$, $-\lambda_nb_n\le b_{n+1}\le-\mu_nb_n$. Furthermore, those sequences satisfy $\lim_n\lambda_n=\lim_n\mu_n=\rho$.
\end{lemma}

In particular, $\lim_{n\to\infty}\frac{b_n}{w_1\rho^n}=1$.

\begin{lemma}\label{lmm:anbn linear independence}
For every $m<n$, there is no $\alpha\in\R$ such that $(a_n,b_n)=\alpha\cdot(a_m,b_m)$, iff $(m,n)\notin\s{(1,3),(4,6),(5,7)}$.
\end{lemma}

We now construct the games that achieve the lower bound. For every $n$ define $B_n\in\{0,1\}^{n\times n}$ to have 1s exactly on the main diagonal, on the diagonal above it, and on the second diagonal below the main one. For example,
\begin{align*}
    B_5=\left(\begin{array}{cccccccccc}
        1&&1&&0&&0&&0\\
        0&&1&&1&&0&&0\\
        1&&0&&1&&1&&0\\
        0&&1&&0&&1&&1\\
        0&&0&&1&&0&&1
    \end{array}\right)
\end{align*}
Define
\begin{align*}
    \beta_n=\left(\begin{array}{c|c}
        {\bf0}_{n-1\times1}&B_{n-1}\\\hline
        1&{\bf0}_{1\times n-1}
    \end{array}\right)\in\Z_2^{n\times n}
\end{align*}
As an example,
\begin{align}
\label{eq:beta10}
    \beta_{11}=\left(\begin{array}{cccccccccccccccccccccccccccc}
    0&&1&&1&&0&&0&&0&&0&&0&&0&&0&&0\\
    0&&0&&1&&1&&0&&0&&0&&0&&0&&0&&0\\
    0&&1&&0&&1&&1&&0&&0&&0&&0&&0&&0\\
    0&&0&&1&&0&&1&&1&&0&&0&&0&&0&&0\\
    0&&0&&0&&1&&0&&1&&1&&0&&0&&0&&0\\
    0&&0&&0&&0&&1&&0&&1&&1&&0&&0&&0\\
    0&&0&&0&&0&&0&&1&&0&&1&&1&&0&&0\\
    0&&0&&0&&0&&0&&0&&1&&0&&1&&1&&0\\
    0&&0&&0&&0&&0&&0&&0&&1&&0&&1&&1\\
    0&&0&&0&&0&&0&&0&&0&&0&&1&&0&&1\\
    1&&0&&0&&0&&0&&0&&0&&0&&0&&0&&0
    \end{array}\right)
\end{align}

The determinant of $\beta_n$ has a nice form:

\begin{lemma}\label{lmm:beta_n det}
$\det B_1=1,\det B_2=1,\det B_3=2$. $\det B_n=\det B_{n-1}+\det B_{n-3}$. Also $\det\beta_n=(-1)^{n+1}\det B_n$.
\end{lemma}

So $\det B_n$ is also a linear recurrence. Similarly to before, its characteristic polynomial is $x^3-x^2-1$, whose roots are $-\rho,-z,-\overline{z}$. As a result a similar analysis yields the same properties for the recurrence $\det B_n$, and $\det\beta_n=\Theta(\rho^n)$. In particular $\beta_n$ is invertible for all sufficiently large $n$.

We now present the two main lemmas, that together yield the $2^{\Omega(n)}$ lower bound.

\begin{lemma}
\label{lmm:beta_n nes}
For every $n\notin\s{3,6,7}$, the game $(I_n,\beta_n)$ has a unique NE and it is fully mixed.
\end{lemma}

\begin{lemma}
\label{lmm:beta_n complexity}
Let $n\ge8$ and let $(x,y)$ be the unique NE of $(I_n,\beta_n)$. Except for a subsequence $m_n=\omega(n)$, there holds $C(x)=2^{\Theta(n)}$, while $C(y)=n$.
\end{lemma}

We now sketch the proof of Lemma \ref{lmm:beta_n complexity}.

\subsubsection{Proof of Lemma \ref{lmm:beta_n complexity}}
First, we need more technical lemmas. Let $(x,y)$ be the fully-mixed NE of $(I_n,\beta_n)$. The first one shows that $K(\beta_n)$ is large.
\begin{lemma}\label{lmm:K(beta_n) bounds}
$|K(\beta_n)|\ge\frac{n}{3}|\det\beta_n|$.
\end{lemma}

Let $g_n=\gcd(a_n,b_n)$. By Lemma \ref{lmm:anbn}, $g_n=\gcd(b_n+b_{n+1},b_n)=\gcd(b_n,b_{n+1})$.
We now characterize $C(x)$ in terms of $\beta_n$ and $g_n$.

\begin{lemma}\label{lmm:complexity of equilibrium}
$C(x)=\frac{|K(\beta_n)|}{g_n}$.
\end{lemma}

The following key lemma establishes an upper bound on $g_n$ that is sufficiently smaller than $|K(\beta_n)|$. This will imply by Lemma \ref{lmm:complexity of equilibrium} that $C(x)=2^{\Omega(n)}$. The proof idea is to lower-bound the number of iterations performed by the Euclidean algorithm on $b_n,b_{n+1}$.

\begin{lemma}\label{lmm:gcd upper bound}
There exists a $0<\gamma<\rho$ such that $g_n\le\gamma^n$ for all sufficiently large $n$, except for a subsequence $m_n=\omega(n)$.
\end{lemma}

We can now finish the proof of Lemma \ref{lmm:beta_n complexity}.
\begin{proof}[Proof of Lemma \ref{lmm:beta_n complexity}]
By Lemmas \ref{lmm:K(beta_n) bounds}, \ref{lmm:complexity of equilibrium} and \ref{lmm:gcd upper bound}, for all sufficiently large $n$, except for the subsequence $m_n$ from Lemma \ref{lmm:gcd upper bound},
\begin{align*}
    |K(\beta_n)|\ge C(x)=\frac{|K(\beta_n)|}{\gcd(b_n,b_{n+1})}\ge\frac{\frac{n}{3}(\sqrt2+0.02)^n}{(\sqrt{2}+0.01)^n}
\end{align*}
So $C(x)=2^{\Theta(n)}$ there. For $y$, observe that $y=U_n$, so the lemma follows.
\qed\end{proof}

\begin{corollary}[Restatement of Theorem \ref{thm:2.4}, Item \ref{item:thm 2.4.1}]\label{thm:final}
There is a function $g(n)=\Theta(n)$ such that for every $n$ there are $A,B\in\Z_2^{n\times n}$ such that $C_1(A,B)=2^{g(n)}$ and $C_2(A,B)=n$.
\end{corollary}
\begin{proof}
Follows by taking the games above and padding them with dummy strategies as in Lemma \ref{lmm:game inflation}.
\qed\end{proof}

The final corollary of this subsection states that the equilibrium distributions in Lemma \ref{lmm:beta_n nes} require the maximal amount of space (asymptotically) to store them.

\begin{corollary}\label{cor:x2 is balanced}
For every $n\ge8$, the NE distribution $x$ of player 1 satisfies $\max_{i,j}\frac{x_i}{x_j}=O(1)$. Therefore representing the numerators of $x$ would require $\Theta(n^2)$ bits.
\end{corollary}

\subsection{Constant-Sum Games}\label{subsec:constant}
Here we prove Theorem \ref{thm:2.3}, Item \ref{item:thm 2.3.2} and Theorem \ref{thm:2.4}, Item \ref{item:thm 2.4.2}. The full proofs can be found in Appendix \ref{appendix:constant}. The lower bounds on the required capability of player 1 remain the same as in the respective Item 1s, although in these constant-sum games, these large minimal required capabilities are now imposed on both players.

We introduce a notion of symmetry of matrices that will help us establish the main results of this subsection. Intuitively, the matrices considered are symmetric with respect to a specific pair of row and column permutations.

\begin{definition}
For $\pi,\tau\in\sym_n$, call a matrix $B$ $(\pi,\tau)$-symmetric if for all $i\in[n]$, $(B^T)_i=\pi(B_{\tau(i)})$.
\end{definition}

The next key lemma shows that if an imitation game satisfies some symmetry conditions, then there is an analogous zero-sum game, where the imitated player's payoff matrix remains unchanged, while the two players' equilibrium strategies become identical, up to permutations.

\begin{lemma}\label{lmm:2p0s unique fully mixed ne lemma}
Let $B\in\Z_2^{n\times n}$ such that 
\begin{enumerate}
    \item it is $(\pi,\tau)$-symmetric,
    \item the game $(I_n,B)$ has a unique NE $(\Tilde{x},U_n)$ which is fully mixed,
    \item $\det B\ne0$ and $K(B)\ne\det B$.
\end{enumerate}
Then $({\bf1}_{n\times n}-B,B)$ has a unique NE which is $(\Tilde{x},\pi^{-1}(\Tilde{x}))$.
\end{lemma}

We note that in any imitation game $(I_n,B)$, every NE $(x,y)$ has the property that $y=U_{\supp y}$. With this key lemma we can establish Item 2 of Theorems \ref{thm:2.3} and \ref{thm:2.4}. We begin with the former.

\begin{theorem}[Restatement of Theorem \ref{thm:2.3}, Item \ref{item:thm 2.3.2}]\label{thm:2.3.2 formal}
There is a function $f(n)=\Tilde\Theta(\sqrt n)$ such that for every $n$ there are $A,B\in\Z_2^{n\times n}$ such that $A+B={\bf1}_{n\times n}$ and $C_1(A,B)=C_2(A,B)=2^{f(n)}$.
\end{theorem}

\begin{theorem}[Restatement of Theorem \ref{thm:2.4}, Item \ref{item:thm 2.4.2}]
There is a function $g(n)=\Theta(n)$ such that for every $n$ there are $A,B\in\Z_2^{n\times n}$ such that $A+B={\bf1}_{n\times n}$ and $C_1(A,B)=C_2(A,B)=2^{g(n)}$.
\end{theorem}
\begin{proof}
Let $n\ge8$ and let $\beta_n$. It is $(\pi,\pi)$-symmetric for $\pi(i)=n-i+1$. By Lemma \ref{lmm:beta_n nes}, $(I_n,\beta_n)$ has a unique NE and it is fully mixed. Also, by Lemmas \ref{lmm:beta_n det} and \ref{lmm:K(beta_n) bounds}, $\det\beta_n\ne0$ and $|K(\beta_n)|\ge\frac{n}{3}|\det\beta_n|$. Therefore, by Lemma \ref{lmm:2p0s unique fully mixed ne lemma}, $({\bf1}_{n\times n}-\beta_n,\beta_n)$ has a unique NE $(\Tilde{x},\pi^{-1}(\Tilde{x}))$ and it is fully mixed. Furthermore, by Lemma \ref{lmm:beta_n complexity}, there holds
\begin{align*}
    C_1({\bf1}_{n\times n}-\beta_n,\beta_n)=C_2({\bf1}_{n\times n}-\beta_n,\beta_n)=C(\Tilde{x})=C(\pi^{-1}(\Tilde{x}))=2^{\Theta(n)}
\end{align*}
Then, a similar argument to Corollary \ref{thm:final} establishes the lower bound for all $n$.
\qed\end{proof}

\begin{credits}
\subsubsection*{Acknowledgements.} We thank Farid Arthaud and Ziyu Zhang for helpful discussions and valuable feedback on earlier versions of the introduction.
\end{credits}

\printbibliography{}

\newpage

\appendix

\section{Omitted Proofs in Section \ref{sec:lemmas}}\label{sec:proofs}
\begin{proof}[Proof of Lemma \ref{lmm:I subseteq J}]
Because $\supp x\subseteq\br_1(y)\subseteq\supp y$. Part 2 follows from the best-response condition \cite{roughgarden2010algorithmic}.
\qed\end{proof}

\begin{proof}[Proof of Lemma \ref{lmm:game inflation}]
Given $A,B\in\Z^{n\times n}$, define
\begin{align*}
    A'=\left(\begin{array}{c|c}
    A&\bf1\\\hline
    {\bf0}&1
    \end{array}\right),B'=\left(\begin{array}{c|c}
    B&{\bf0}\\\hline
    \bf1&0
    \end{array}\right)
\end{align*}
Then $\NE(A,B)=\NE(A',B')$ follows (with abuse of notation regarding the equilibrium distributions).
\qed\end{proof}

\begin{proof}[Proof of Lemma \ref{lmm:1}]
Let $x=A^{-1}\bf1$. We have $n=\inp{Ax,1}=\sum_{i,j}A_{i,j}x_i$.
Because $A$ is binary and invertible, $\sum_jA_{i,j}\in[n]$ for all $i$. Using Cramer's rule \cite{cramer1750introduction},
\begin{align*}
    x_j=\frac{\det(A\stackrel{i}{\gets}\bf1)}{\det A}~\Rightarrow~\sum_jx_j=\frac{K(A)}{\det A}\ge0
\end{align*}
Therefore
\begin{align*}
    \frac{K(A)}{\det A}\le n=\sum_{i,j}A_{i,j}x_i\le n\frac{K(A)}{\det A}
\end{align*}
Hence $|\det A|\le|K(A)|\le n|\det A|$.
\qed\end{proof}

The proof also implies that under the conditions above (regardless of whether $A$ is binary), $\det A,K(A)$ have the same sign.
\begin{proof}[Proof of Lemma \ref{lmm:2}]
The main inequality in the previous proof becomes
\begin{align*}
    n=\sum_{i,j}A_{i,j}x_i\le M\frac{K(A)}{\det A}
\end{align*}
Hence $|K(A)|\ge\frac{n}{M}|\det A|$.
\qed\end{proof}

\begin{proof}[Proof of Theorem \ref{thm:K upper bound}]
Let $x=(B^T)^{-1}{\bf1},y=A^{-1}\bf1$. By the support enumeration algorithm (see \cite{roughgarden2010algorithmic}, Chapter 3), $x,y\ge\bf0$ and the unique fully-mixed NE of the game is $\p{\frac{x}{\norm{x}_1},\frac{y}{\norm{y}_1}}$. As before, we have
\begin{align*}
    \norm{x}_1=\sum_ix_i=\frac{K(B^T)}{\det B^T}=\frac{K(B)}{\det B}
\end{align*}
($K(B^T)=K(B)$ by definition). Therefore
\begin{align*}
    \frac{x_i}{\norm{x}_1}=\frac{\det(B^T\stackrel{i}{\gets}\bf1)}{K(B)}
\end{align*}
Since $B$ has integer entries, the numerator and denominator of the fraction above are integers, so we get $C_1(A,B)\le C\p{\frac{x}{\norm{x}_1}}\le K(B)$. An analogous argument shows the second inequality.
\qed\end{proof}

\section{Proof of Theorem \ref{thm:upper bound binary}}\label{appendix:proof upper bound}
\begin{proof}[Proof of Theorem \ref{thm:upper bound binary}]
Let $(x,y)\in\NE(A,B)$ with supports $I,J$. Let $k=|I|,l=|J|$.

Write $x=\p{\frac{p_1}{q},...,\frac{p_n}{q}}$ where $p_i,q\in\N$ and $\gcd(p_1,...,p_n)=1$. Consider $p_I:=(p_i)_{i\in I},B,J$. Then they satisfy the linear equalities and inequalities
\begin{align*}
    (B_{I,J})^Tp_I=u\cdot{\bf1}_l,~~~(B_{I,[n]\setminus J})^Tp_I\le u\cdot{\bf1}_{n-l},~~~p_I\ge{\bf1}_k,~~~(p_I,u)\in\N_{>0}^{k+1}
\end{align*}
In words, these equalities and inequalities mean that conditioned on player 1 playing $x$, the pure strategies in $J$ all yield the maximal utility for player 2, and that player 1 plays every pure strategy $i\in I$ with positive probability (see \cite{roughgarden2010algorithmic}, Chapter 3 for further details).
We reorganize terms to get the equivalent problem
\begin{align*}
    (B_{I,J})^Tp_I-u\cdot{\bf1}_l={\bf0}_l,~~\left(\begin{array}{c}u\cdot{\bf1}_{n-l}-(B_{I,[n]\setminus J})^Tp_I\\p_I\end{array}\right)\ge\left(\begin{array}{c}{\bf0}_{n-l}\\{\bf1}_k\end{array}\right),~~(p_I,u)\in\N_{>0}^{k+1}
\end{align*}
So to minimize $q+u=\sum_{i\in I}p_i+u$ where $(p_I,u)$ is the variable, \citet{von1978bound} show that there exists a solution $(p_I,u)$ where
\begin{align*}
    \sum_{i\in I}p_i+u\le n(n+1)M
\end{align*}
where $M$ is an upper bound on the absolute values of all subdeterminants of the matrix
\begin{align*}
    \left(\begin{array}{c|c}
        \begin{array}{c|c}(B_{I,J})^T&-{\bf1}_l\end{array}&{\bf0}_l\\\hline
        \left(\begin{array}{c}
        \begin{array}{c|c}-(B_{I,[n]\setminus J})^T&{\bf1}_{n-l}\end{array}\\\hline I_{k+1}\end{array}\right)&\left(\begin{array}{c}{\bf0}_{n-l}\\\hline{\bf1}_{k+1}
        \end{array}\right)
    \end{array}\right)
\end{align*}
This is an $n+k+1\times k+2$ matrix, so by \citet{hadamard1893resolution} we can bound $M\le\p{\max_{i,j}\abs{B_{i,j}}}^n\cdot(2n+1)^{\frac{2n+1}{2}}$ and then $q=\sum_ip_i\le\p{\max_{i,j}\abs{B_{i,j}}}^n\cdot2^{O(n\logn)}$. Therefore $C_1(A,B)\le q\le\p{\max_{i,j}\abs{B_{i,j}}}^n\cdot2^{O(n\logn)}$. A similar inequality follows for $C_2$.
\qed\end{proof}

\section{Omitted Proofs in Section \ref{subsec:2^sqrt n}}\label{appendix:proofs 2^sqrt n}
\begin{lemma}\label{lmm:det B^k}
$\det B^k=k$. Let $B'$ result from $B^k$ by replacing one of its 0-entries with a 1. Then $\det B'=1$.
\end{lemma}
\begin{proof}
$B^k$ is a $k+1\times k+1$ matrix. Let $\chi_A$ be the characteristic polynomial of a matrix $A$. For the matrix ${\bf1}_{k+1\times k+1}$, we have $\chi_{{\bf1}_{k+1\times k+1}}(x)=(x-k-1)x^k$. Therefore
\begin{align*}
    \det(B^k)&=(-1)^k\det({\bf1}_{k+1\times k+1}-I_{k+1})=-\det(I_{k+1}-{\bf1}_{k+1\times k+1})\\
    &=-\chi_{{\bf1}_{k+1\times k+1}}(1)=k
\end{align*}
where the first equality is by applying $k$ row swaps.
Let $B'$. Apply the same row swaps that transform $B'$ into ${\bf1}_{k+1\times k+1}-I_{k+1}$ up to one main-diagonal entry (where the entry is 1). Suppose that that entry is $(i,i)$. Then by the above, when expanding the determinant along row $i$,
\begin{align*}
    \det B'&=k+(-1)^k(-1)^{2i}\det({\bf1}_{k+1\times k+1}-I_{k+1})[i|i]\\
    &=k+(-1)^k\det({\bf1}_{k\times k}-I_k)=k-(k-1)=1
\end{align*}
by definition.
\qed\end{proof}

\begin{proof}[Proof of Theorem \ref{thm:2^sqrt n lower bound}]
For $n$ let $p_1,...,p_n$ be the first $n$ primes. Let $N=\sum_{k\le n}(p_k+1)+1=\sum_kp_k+n+1$. By the prime number theorem \cite{apostol1998introduction}, $p_k=\Theta(k\log k)$, so $N=\Theta(n^2\logn)$. Define the matrix $B'=diag(B^{p_1},B^{p_2},...,B^{p_n})\in\Z_2^{N-1\times N-1}$, and then define
\begin{align}
    B=\left(\begin{array}{c|c}
        {\bf0}_{N-1\times1}&B'\\\hline
        1&{\bf0}_{1\times N-1}
    \end{array}\right)\in\Z_2^{N\times N}
\end{align}
We first prove that $(A,B)$ has only one NE and it is fully mixed. Let $(x,y)\in\NE$ with supports $I,J\subseteq[N]$. Suppose $I\ne[N]$. Then for some $u$, $x$ satisfies
\begin{align*}
    \forall j\in J,~~~\sum_{i\in I}x_iB_{ij}=\sum_{i,~B_{ij}=1}x_i=u
\end{align*}
Here we abuse notation and refer to $B^{p_k}$ also as its set of row indices in $B$. For $i\in B^{p_k}$ we also write $i+1\bmod B^{p_k}$ to refer to either $i+1$ if $i<\max B^{p_k}$ or $\min B^{p_k}$ otherwise.
First, suppose that $k\le n$ is such that $I\cap B^{p_k}\notin\s{\emptyset,B^{p_k}}$. If $|I\cap B^{p_k}|>1$, then linear algebra tells us that $x_i=0$ for every $i\in B^{p_k}$, contradicting the assumption on $I$. Otherwise, $I\cap B^{p_k}=\{i\}$. Then, $B_{ii}=0$, and because $i>1$ we get $(x^TB)_i=0$, while $(x^TB)_{i+1\bmod B^{p_k}}>0$. Therefore, $i\notin J$, so because $A=I_N$ we get $i\notin I$, a contradiction. Therefore $I\cap B^{p_k}\in\s{\emptyset,B^{p_k}}$.
Now focus on $i=N$. Suppose that $i\notin I$, and let $i_1=\min I$. By the definition of $B$ (specifically because $B'$ is placed off-diagonally), and using the previous paragraph, it follows that $(x^TB)_j=0$ for all $j\in\s{1,...,i_1-1,i_1,N}$, while $(x^TB)_j>0$ for some other $j$. Therefore $i_1\notin J$, so $i_1\notin I$ -- a contradiction. In addition, clearly there holds $I\setminus\{N\}\ne\emptyset$.
It is left to show that $B^{p_k}\subseteq I$ for all $k$. Suppose that $B^{p_k}\cap I=\emptyset$. Let $i=\max B^{p_k}$. Using a similar argument, we get $y_{i+1}=0$, so $x_{i+1}=0$, but $i+1\in I$, because $N\in I$. Hence $I=[N]$. By Lemma \ref{lmm:I subseteq J}, $J=[N]$.

Now we find the fully-mixed NE. Its uniqueness follows by the fact that $\det I_n,\det B\ne0$. Let $(x,y)$ be the NE. For $y$, because $A=I_N$ one immediately gets $y=U_N$. For some $u>0$, $x$ needs to satisfy $B^Tx=\bf1\cdot u$. Note that $B'$ is a block-diagonal matrix, and when expanding along the first column of $B$,
\begin{align*}
    \det B=(-1)^{N-1}\prod_{k=1}^n\det(B^{p_k})\stackrel{\text{Lemma \ref{lmm:det B^k}}}{=}(-1)^{N-1}\prod_{k=1}^np_k
\end{align*}
So by Cramer's rule we get
\begin{align*}
    \forall i\le N~~~x_i=u\frac{\det(B^T\stackrel{i}{\gets}{\bf1}_N)}{\det B^T}
\end{align*}
Let $k$ such that row $i$ goes through block $B^{p_k}$. After replacing the blocks we get an upper-triangular matrix, where the diagonal contains the blocks
$$B^{p_1},...,B^{p_{k-1}},B^{p_{k+1}},...,B^{p_n},$$
and the block $B^{p_k}$ with one of its 0-entries replaced by a 1. Call the last block $\beta$. So we get
\begin{align*}
    &\det((B\stackrel{i}{\gets}{\bf1}_N)^T)\\
    &=(-1)^{N-1}\cdot((-1)^{(p_k+1)\sum_{l<k}(p_l+1)}(-1)^{(p_k+1)\sum_{l<k}(p_l+1)})\det\beta\cdot\prod_{l\ne k}p_l\\
    &=(-1)^{N-1}\det\beta\cdot\prod_{l\ne k}p_l\stackrel{\text{Lemma \ref{lmm:det B^k}}}{=}(-1)^{N-1}\prod_{l\ne k}p_l=\frac{\det B}{p_k}
\end{align*}
Therefore, $x_i=\frac{u}{p_k}$. If $i=1$ we can define $p_k=1$ and everything still holds. To normalize, we compute $\inp{x,{\bf1}_N}$:
\begin{align*}
    \inp{x,{\bf1}_N}=u\p{1+\sum_{k=1}^n\frac{p_k+1}{p_k}}=u\p{n+1+\sum_{k=1}^n\frac{1}{p_k}}
\end{align*}
Hence $u=n+1+\sum_{k=1}^n\frac{1}{p_k}$, and the normalized $x$ becomes $x_i=\frac{1}{p_k\p{n+1+\sum_{l=1}^n\frac{1}{p_l}}}$.

Now, observe that $x_N=\frac{\prod_kp_k}{\p{\prod_kp_k}\p{n+1+\sum_{k=1}^n\frac{1}{p_k}}}$ is $x_N$ written as a fraction where the numerator and denominator are integers. The same holds for every $k$ and the fractions $x_i=\frac{\prod_lp_l}{p_k\prod_lp_l\p{n+1+\sum_{l=1}^n\frac{1}{p_l}}}$. Therefore $C_1(A,B)\le\p{\prod_kp_k}\p{n+1+\sum_{k=1}^n\frac{1}{p_k}}$. We now claim that this is the reduced form for $x_N$, establishing that $C_1(A,B)=\p{\prod_kp_k}\p{n+1+\sum_{k=1}^n\frac{1}{p_k}}$. If not, then because the fraction is $\le1$, it follows that there exists an $\alpha\in\N$ such that $\alpha\prod_kp_k=\p{\prod_kp_k}\p{n+1+\sum_k\frac{1}{p_k}}$, which means that $n+1+\sum_k\frac{1}{p_k}\in\N$, so $\sum_k\frac{1}{p_k}\in\N$. But this is false: for $n=1$, $\frac{1}{p_1}=\frac{1}{2}\notin\N$. If it was true for $n>1$ then it would imply that $\frac{1}{p_n}+\sum_{k<n}\frac{1}{p_k}\in\N$. Since $p_n$ is prime, this means that $p_n\mid\prod_{k<n}p_k$, which is a contradiction as they are all distinct primes.

Now we show the asymptotics of $C_1(A,B)$. First, $\log(\prod_kp_k)=\sum_k\log p_k=\Theta(p_n)=\Theta(n\logn)$ as this is the first Chebyshev function \cite{apostol1998introduction}   , and using the prime number theorem. Therefore $\prod_kp_k=2^{\Theta(n\logn)}=2^{\Tilde{\Theta}(n)}$. Since $\sum_k\frac{1}{p_k}=\Theta(\log\log p_n)=\Theta(\loglogn)$, we get $C_1(A,B)=2^{\Theta(n\logn)}=2^{\Tilde{\Theta}(\sqrt N)}$.
\qed\end{proof}

\begin{proof}[Proof of Corollary \ref{cor:x1 is balanced}]
By definition and the prime number theorem, it is $p_{\Tilde\Theta(\sqrt n)}=\Tilde\Theta(\sqrt n\log\sqrt n)=\Tilde\Theta(\sqrt n)$. This implies that all numerators are at least $\frac{2^{\Tilde\Theta(\sqrt n)}}{n\Tilde\Theta(\sqrt n)}$, so they would take up together $\Tilde\Theta(n\sqrt n)$ bits.
\qed\end{proof}

\section{Omitted Proofs in Section \ref{subsec:2^n}}\label{appendix:proofs 2^n}

\begin{lemma}\label{lmm:bn sign}
For $n\ge4$, $\sign(b_n)=(-1)^n$. In particular, $b_n\ne0$ for every $n\ge5$.
\end{lemma}
\begin{proof}
$\sign(b_4)=\sign(1)=1$. By induction, for $n>4$, $\sign(b_n)=\sign(b_{n-2}-b_{n-3}+b_{n-4})=\sign(b_{n-2})=(-1)^n$ because all 3 signed terms have the same sign (or some are 0 but not all).
\qed\end{proof}

\begin{proof}[Proof of Lemma \ref{lmm:anbn}]
It is clear for $n\le4$. By induction, for $n\ge5$,
\begin{align*}
    a_n&=a_{n-2}-a_{n-3}+a_{n-4}=b_{n-2}+b_{n-1}-b_{n-3}-b_{n-2}+b_{n-4}+b_{n-3}\\
    &=b_n+b_{n+1}
\end{align*}
\qed\end{proof}

\begin{proof}[Proof of Lemma \ref{prop:ratio sequences to rho}]
Note that $|w_2|<\frac{1}{3}$. So
\begin{align*}
    |w_2z^n+\overline{w_2}\cdot\overline{z}^n|=|w_2z^n+\overline{w_2z^n}|=|2\Re(w_2z^n)|\le2|w_2z^n|<\frac{2}{3}|z|^n<|\rho|^n
\end{align*}
Therefore
\begin{align*}
    \lim_{n\to\infty}\frac{b_{n+1}}{b_n}=\lim_{n\to\infty}\frac{\frac{1}{3}+w_1\rho^{n+1}+2\Re(w_2z^{n+1})}{\frac{1}{3}+w_1\rho^n+2\Re(w_2z^n)}=\frac{w_1\rho}{w_1}=\rho
\end{align*}
In addition, one can see from the expression for $\frac{b_{n+1}}{b_n}$ that (for $n\ge5$) (i) its distance to $\rho$ is strictly decreasing; (ii) $\frac{b_{n+1}}{b_n}>\rho$ for $n\in2\N$ and $<\rho$ otherwise. The lemma statement now follows.
\qed\end{proof}

Lemma \ref{prop:ratio sequences to rho}, $|\rho|>1$ and Lemma \ref{lmm:anbn} imply that
\begin{proposition}
For every $n\ge6$, $a_n\ne0$.
\end{proposition}

\begin{proof}[Proof of Lemma \ref{lmm:anbn linear independence}]
Let $6\le m<n$ and assume equality for contradiction. Then $a_m,b_m,a_n,b_n\ne0$ so $\alpha\ne0$. Equivalently we have $\frac{a_n}{a_m}=\frac{b_n}{b_m}$. So
\begin{align*}
    a_nb_m=a_mb_n&\iff(b_n+b_{n+1})b_m=(b_m+b_{m+1})b_n\iff b_{n+1}b_m=b_nb_{m+1}\\
    &\iff\frac{b_{n+1}}{b_n}=\frac{b_{m+1}}{b_m}
\end{align*}
Call the last two ratios $\rho_n,\rho_m$. If $n-m\notin2\N$ then either $\rho_n<\rho<\rho_m$ or $\rho_m<\rho<\rho_n$. If $n-m\in2\N$ then either $\rho_m<\rho_n<\rho$ or $\rho<\rho_n<\rho_m$. Both cases follow from Lemma \ref{prop:ratio sequences to rho}, and both imply a contradiction.

One can verify the correctness of the claim when $m\le 5$.
\qed\end{proof}

\begin{proof}[Proof of Lemma \ref{lmm:beta_n det}]
Parts 1 and 3 are immediate. For part 2:
\begin{align*}
    \det B_n&=\det B_{n-1}-\det(B_n[1|2])=\det B_{n-1}+\det((B[1|2])[2|1])\\
    &=\det B_{n-1}+\det(((B[1|2])[2|1])[1|1])=\det B_{n-1}+\det B_{n-3}
\end{align*}
by definition.
\qed\end{proof}

\subsubsection{Proof of Lemma \ref{lmm:beta_n nes}}
The lemma follows immediately from the following:
\begin{lemma}
For every $n\ge1$,
\begin{align*}
    \s{\supp x\mid(x,y)\in\NE(I_n,\beta_n)}&=\s{[n]\setminus\s{i\le n\mid\exists\alpha~(a_{n+1},b_{n+1})=\alpha\cdot(a_i,b_i)}},\\
    \s{\supp y\mid(x,y)\in\NE(I_n,\beta_n)}&=\s{J\supseteq\supp x\mid(x,y)\in\NE(I_n,\beta_n)}
\end{align*}
\end{lemma}
\begin{proof}
Since in this case, by Lemma \ref{lmm:anbn linear independence},
$$|\s{[n]\setminus\s{i\le n\mid\exists\alpha~(a_{n+1},b_{n+1})=\alpha\cdot(a_i,b_i)}}|=1$$
for all $n\ge8$, we assume that $n\ge10$ and simply show that there is one NE and it is fully mixed. The cases $n\le9$ can be verified by hand. Let $(x,y)\in\NE$ with supports $I,J$.
Observe that $I\subseteq J$ because $A=I_n$ and $x\in\Delta\br_1(y)$. 
We first show that $\s{1,n-4,n-3,...,n}\subseteq I$. Let $i\in\s{1,n-4,n-3,...,n}$, and suppose for contradiction that $i\notin I$. We refer to Equation \ref{eq:beta10} to help visualize this.

\paragraph{Case $i=1$:} Let $i_1=\min I$.
\begin{enumerate}
    \item If $i_1=2$, then $3\in I$ (otherwise $2\notin J$). Then, $(x^T\beta_n)_2<(x^T\beta_n)_4$, so $2\notin J$, which is a contradiction.
    \item If $3\le i_1\le n-2$, then $i_1+1\notin I$ (otherwise $i_1\notin J$).
    \item If $i_1=n-1$ then $J\subseteq\s{1,n-2,n}$ -- a contradiction.
    \item If $i_1=n$ then $J=\s{1}$ -- a contradiction.
\end{enumerate}

\paragraph{Case $i=n$:} Then $(x^T\beta_n)_1=0$, so $1\notin I$, so it is again a contradiction.

Therefore $1,n\in I\subseteq J$. To not contradict that we have $\s{n-2,n-1}\cap I\ne\emptyset$.

\paragraph{Case $i=n-1$:} Then $n-2\in I$. Then $\s{n-4,n-3}\cap I\ne\emptyset$. If $n-3\in I$ then $(x^T\beta_n)_n<(x^T\beta_n)_{n-1}$ and if $n-4\in I$ then $(x^T\beta_n)_n<(x^T\beta_n)_{n-3}$, so $n\notin J$ -- a contradiction.

\paragraph{Case $i=n-2$:} Then $n-3\in I$ (otherwise $(x^T\beta_n)_{n-1}=0$ and then $n-1\notin J$ but $n-1\in I$). Then, $(x^T\beta_n)_n<(x^T\beta_n)_{n-2}$, so $n\notin J$ -- a contradiction.

\paragraph{Case $i=n-3$:}
\begin{enumerate}
    \item If $n-2\notin I$ then $n-1\in I$ and this implies that $n-3\in I$ -- a contradiction.
    \item If $n-1\notin I$ then $n-2\in I$ and this implies that $n-4\in I$, but then $(x^T\beta_n)_n<(x^T\beta_n)_{n-3}$, so $n\notin J$ -- a contradiction.
    \item If $n-2,n-1\in I$ then $(x^T\beta_n)_{n-1}<(x^T\beta_n)_n$ so $n-1\notin J$ -- a contradiction.
\end{enumerate}

\paragraph{Case $i=n-4$:} By the above, $\s{n-3,n-2,n-1,n}\subseteq I$. With that, and looking at $(x^T\beta_n)_{n-2,n-1,n}$, we get that $x_{n-3}=x_{n-2}=x_{n-1}$. Then, because $n-3\in I\subseteq J$ (Lemma \ref{lmm:I subseteq J}), there holds $n-5\in I$ and $x_{n-5}=x_{n-3}=...=x_{n-1}$. Then, we also get $n-6\notin I$, as otherwise $(x^T\beta_n)_{n-4}>(x^T\beta_n)_n$.
Similarly, $n-5\in I$ implies that $n-7\in I$. Then, $(x^T\beta_n)_{n-5}=(x^T\beta_n)_{n-3}$, namely $x_{n-7}=2x_{n-5}$. But then $(x^T\beta_n)_{n-6}\ge x_{n-5}+x_{n-7}=3x_{n-5}>(x^T\beta_n)_n$, in contradiction to $n\in J$.
\\\\
Now suppose that $J=[n]$. We show that $I=[n]$.
$n-4,...,n\in I$, so $x_{n-4},...,x_n>0$ and $n-4,...,n\in J$, so $(x^T\beta_n)_{n-4}=...=(x^T\beta_n)_n$.
For every $k$ express each $x_{n-k}$ as a tuple $(a'_k,b'_k)$ of $x_{n-1},x_{n-4}$, so that $x_{n-k}=a'_kx_{n-1}+b'_kx_{n-4}$. Write $(a'_k,b'_k)=0$ to mean that $a'_kx_{n-1}+b'_kx_{n-4}=0$.
We have
\begin{align*}
    (x^T\beta_n)_n=x_{n-1}+x_{n-2}=x_{n-2}+x_{n-3}=(x^T\beta_n)_{n-1}\Rightarrow x_{n-1}=x_{n-3}
\end{align*}
Also,
\begin{align*}
    &(x^T\beta_n)_n=x_{n-1}+x_{n-2}=x_{n-1}+x_{n-3}+x_{n-4}=(x^T\beta_n)_{n-2}\\
    &\Rightarrow x_{n-2}=x_{n-1}+x_{n-4}
\end{align*}
So we get the matching $x_{n-1},...,x_{n-4}=(1,0),(1,1),(1,0),(0,1)$ so $(a'_k,b'_k)=(a_k,b_k)$ for $k\le4$.
For every $5\le k\le n-1$, because $J=[n]$ we get
\begin{align*}
    &(x^T\beta_n)_{n-k+2}=x_{n-k}+x_{n-k+1}+x_{n-k+3}\\
    &=x_{n-k+1}+x_{n-k+2}+x_{n-k+4}=(x^T\beta_n)_{n-k+3}\\
    &\Rightarrow x_{n-k}=x_{n-k+2}-x_{n-k+3}+x_{n-k+4}
\end{align*}
So $(a'_k,b'_k)=(a'_{k-2}-a'_{k-3}+a'_{k-4},b'_{k-2}-b'_{k-3}+b'_{k-4})$. Therefore $(a'_k,b'_k)=(a_k,b_k)$ for every $k\in[n-1]$. Then, comparing columns $2,3$ of $\beta_n$, we get
\begin{align*}
    &(x^T\beta_n)_2=x_1+x_3=x_1+x_2+x_4=(x^T\beta_n)_3\\
    &\Rightarrow0=x_2-x_3+x_4
\end{align*}
Or in terms of $a_k,b_k$, we have $(a_{n-2}-a_{n-3}+a_{n-4},b_{n-2}-b_{n-3}+b_{n-4})=(a_n,b_n)=0$ by definition. Therefore, because $a_n,b_n\ne0$, $\frac{x_{n-1}}{x_{n-4}}=\frac{-b_n}{a_n}$.

Assume for contradiction that $I\ne[n]$. Let $l$ be the minimal $l$ such that $n-l\notin I$. Let $k<l$. Then
\begin{align*}
    &(x^T\beta_n)_{n-l+2}=x_{n-l}+x_{n-l+1}+x_{n-l+3}\\
    &=x_{n-l+1}+x_{n-l+2}+x_{n-l+4}=(x^T\beta_n)_{n-l+3}\\
    &\Rightarrow0=x_{n-l}=x_{n-l+2}-x_{n-l+3}+x_{n-l+4}
\end{align*}
Using the vector expressions for the above and the recursion equations, we get $a_lx_{n-1}+b_lx_{n-4}=0$ so $\frac{x_{n-1}}{x_{n-4}}=-\frac{b_l}{a_l}$.

In total, $\frac{b_n}{a_n}=\frac{b_l}{a_l}$, in contradiction to Lemma \ref{lmm:anbn linear independence}.\\\\
Now suppose for contradiction that $J\subset[n]$. Let $n-l=\min[n]\setminus J$. If $n-l\in\br_2(x)$, then the argument above can be repeated to reach a contradiction to Lemma \ref{lmm:anbn linear independence}, so suppose that $n-l\notin\br_2(x)$. $I\subseteq J$ therefore $x_{n-l}=0$. If $x_{n-l-1}=0$, then $(x^T\beta_n)_{n-l+1}<(x^T\beta_n)_{n-l+3}$, so $n-l+1\notin J$ which is a contradiction. Therefore $n-l-1\in I\subseteq J$.
Now, the same argument as above shows that $l\ge5$ and that for all $k\le l+1$, $x_{n-k}=(a_k,b_k)$. We have the inequality $(x^T\beta_n)_{n-l}<(x^T\beta_n)_{n-l+1}=(2,1)$. Hence
\begin{align*}
    (a_{l+1}+a_{l-1})x_{n-1}+(b_{l+1}+b_{l-1})x_{n-4}+x_{n-l-2}&<2x_{n-1}+x_{n-4},\\
    a_lx_{n-1}+b_lx_{n-4}&=0
\end{align*}
We show that $(a_{l+1}+a_{l-1}-2)x_{n-1}+(b_{l+1}+b_{l-1}-1)x_{n-4}\ge0$ in contradiction.
Using Lemma \ref{lmm:anbn}, this equals $(b_{l+1}+b_{l+2}+b_{l-1}+b_l-2)x_{n-1}+(b_{l+1}+b_{l-1}-1)x_{n-4}$.
The equation above gives us $(b_l+b_{l+1})x_{n-1}=-b_lx_{n-4}$, so using that and the definition of $b_l$ we get
\begin{align*}
    &=(b_{l-1}+b_{l+2}-2)x_{n-1}+(b_{l+3}-1)x_{n-4}\\
    &=(b_{l-1}+b_{l+2}-2)x_{n-1}-\frac{b_l+b_{l+1}}{b_l}(b_{l+3}-1)x_{n-1}\ge0\\
    \iff&b_{l-1}+b_{l+2}-2-\frac{b_l+b_{l+1}}{b_l}(b_{l+3}-1)\\
    &=b_{l-1}+b_{l+2}-1-b_{l+3}-\frac{b_{l+1}}{b_l}(b_{l+3}-1)\ge0
\end{align*}
Expanding $b_{l+3}$ back yields
\begin{align*}
    &b_{l-1}+b_{l+2}-1-(b_{l+1}-b_l+b_{l-1})-\frac{b_{l+1}}{b_l}(b_{l+1}-b_l+b_{l-1}-1)\\
    &=b_{l+2}+b_l-\frac{b_{l+1}}{b_l}(b_{l+1}+b_{l-1})-1+\frac{b_{l+1}}{b_l}\ge0\\
    &\iff b_{l+2}+b_l-\frac{b_{l+1}}{b_l}(b_{l+1}+b_{l-1})\ge1-\frac{b_{l+1}}{b_l}
\end{align*}
Because $l\ge5$, it is sufficient to show that $b_{l+2}+b_l-\frac{b_{l+1}}{b_l}(b_{l+1}+b_{l-1})\ge3$ for all $l\ge5$. We deduce this from the following claim (see proof in Appendix \ref{sec:proofs}).

\begin{claim}\label{claim}
$\lim_{l\to\infty}b_l-\frac{b_{l+1}b_{l-1}}{b_l}=-\frac{(\rho-1)^2}{3\rho}\approx1.38263$.
\end{claim}
\begin{proof}
$b_l-\frac{b_{l+1}b_{l-1}}{b_l}=\frac{b_l^2-b_{l+1}b_{l-1}}{b_l}$. Let $t_l=2\Re(w_2z^l)$, so that $b_l=\frac{1}{3}+w_1\rho^n+t_l$. Using Equation \ref{eq:bn},
\begin{align*}
    b_l^2&=\frac{1}{9}+w_1^2\rho^{2l}+t_l^2+\frac{2w_1}{3}\rho^l+\frac{2}{3}t_l+2w_1\rho^lt_l,\\
    b_{l+1}b_{l-1}&=\frac{1}{9}+w_1^2\rho^{2l}+t_{l+1}t_{l-1}+\frac{w_1}{3}(\rho^{l+1}+\rho^{l-1})+\frac{1}{3}(t_{l+1}+t_{l-1})+w_1(\rho^{l+1}t_{l-1}+\rho^{l-1}t_{l+1})
\end{align*}
So
\begin{align*}
    b_l^2-b_{l+1}b_{l-1}=t_l^2-t_{l+1}t_{l-1}-\frac{w_1}{3}\rho^{l-1}(\rho-1)^2+\frac{1}{3}(2t_l-t_{l+1}-t_{l-1})+w_1\rho^{l-1}(2\rho t_l-\rho^2t_{l-1}-t_{l+1})
\end{align*}
By what we have shown, $|t_l|\le\frac{2}{3}|z|^l$, and $|z|<0.9$. Therefore $t_l\to0$ as $l\to\infty$. Furthermore,
\begin{align*}
    |w_1\rho^{l-1}(2\rho t_l-\rho^2t_{l-1}-t_{l+1})|\le|w_1\rho^{l-1}z^{l-1}|(|\rho|+|z|)^2=:u|\rho^{l-1}||z^{l-1}|
\end{align*}
Also $\frac{u|\rho^{l-1}||z^{l-1}|}{|\rho^{l-1}|}=\frac{u\rho^{l-1}|z^{l-1}|}{\rho^{l-1}}\to0$ when $l\to\infty$. In total, after omitting the negligible terms from the numerator and denominator, we get
\begin{align*}
    \lim_{l\to\infty}\frac{b_l^2-b_{l+1}b_{l-1}}{b_l}=\lim_{l\to\infty}\frac{-\frac{w_1}{3}\rho^{l-1}(\rho-1)^2}{w_1\rho^l}=-\frac{(\rho-1)^2}{3\rho}
\end{align*}
\qed\end{proof}

As a result,
\begin{align*}
    \lim_{l\to\infty}b_{l+2}-\frac{b_{l+1}^2}{b_l}=\lim_{l\to\infty}\frac{\frac{w_1}{3}\rho^l(\rho-1)^2}{w_1\rho^l}=\frac{(\rho-1)^2}{3}\approx2.02635
\end{align*}

It can be now easily verified that $\p{b_l-\frac{b_{l+1}b_{l-1}}{b_l}}+\p{b_{l+2}-\frac{b_{l+1}^2}{b_l}}\ge3$ for all $l\ge5$. Therefore we reach a contradiction, hence $J=[n]$, concluding the proof.
\qed\end{proof}

\begin{corollary}
For every $n\ne3,6,7$, the game $(I_n,\beta_n)$ has a unique NE and it is fully mixed.
\end{corollary}
Note that the proof above also tells us how to construct a fully-mixed NE of $(I_n,\beta_n)$: First, $y=U_n$. As for $x$, first write $x'_{n-1}=|b_n|,x'_{n-4}=|a_n|$. Then define $x'_{n-k}=a_k|b_n|+b_k|a_n|$ for every $k\in\s{2,3,5,6,7,...,n-1}$. Then define $x'_n=2|b_n|+|a_n|$. Then define $x=\frac{x'}{\inp{x,{\bf1}}}$. The uniqueness of the NE is due to $\det\beta_n,\det I_n\ne0$.

\subsubsection{Proof of Lemma \ref{lmm:beta_n complexity}}
\begin{proof}[Proof of Lemma \ref{lmm:K(beta_n) bounds}]
Follows from Lemma \ref{lmm:K det inequality strengthened}.
\qed\end{proof}

\begin{lemma}
$\rho\notin\Q$.
\end{lemma}
\begin{proof}
$-\rho$ is the real root of the polynomial $x^3+x^2+1$. The polynomial evaluated at $2$ is $13$ which is prime, so by Cohn's irreducibility criterion \cite{brillhart1981irreducibility} the polynomial is irreducible in $\Q[x]$, so $\rho\notin\Q$.
\qed\end{proof}

\begin{lemma}\label{lmm:det beta_n bn identity}
For all $n\ge8$, $|\det\beta_n|=2|b_n|+|a_n|$.
\end{lemma}
\begin{proof}
$|\det\beta_n|=\det B_{n-1}=\det B_{n-2}+\det B_{n-4}$. So by induction (with the cases $n\le11$ easy to verify),
\begin{align*}
    \det B_{n-2}+\det B_{n-4}=2|b_{n-1}|+|a_{n-1}|+2|b_{n-3}|+|a_{n-3}|
\end{align*}
Using Lemma \ref{lmm:anbn}, and assuming wlog that $b_{n-1}>0$,
\begin{align*}
    &=2b_{n-1}-(b_{n-1}+b_n)+2b_{n-3}-(b_{n-3}+b_{n-2})\\
    &=b_{n-1}-b_n+b_{n-3}-b_{n-2}=b_{n+1}-b_n\\
    &=a_n-2b_n=|a_n|+2|b_n|
\end{align*}
proving the lemma for all $n\ge8$.
\qed\end{proof}

\begin{proof}[Proof of Lemma \ref{lmm:complexity of equilibrium}]
Assume wlog that $b_n>0$. Write $x=\p{\frac{p_1}{q},...,\frac{p_n}{q}}$ that satisfy $p_1,...,p_n,q\in\N$ and $\mathrm{lcm}(p_1,...,p_n)=q$. Consider the proof of Lemma \ref{lmm:beta_n nes}. We showed that $\frac{x_{n-1}}{x_{n-4}}=-\frac{b_n}{a_n}$, and that $x_{n-k}=a_kx_{n-1}+b_kx_{n-4}$. So by picking $b_n,-a_n$ and dividing by $g_n$ we get
\begin{align*}
    p_{n-1}=\frac{b_n}{g_n},~p_{n-4}=\frac{-a_n}{g_n}
\end{align*}
Then, applying the equality $x_{n-k}=a_kx_{n-1}+b_kx_{n-4}$ to the numerators of the components of $x$, we have
\begin{align*}
    p_{n-k}=\frac{a_kb_n-b_ka_n}{g_n}>0
\end{align*}
and $p_n=2p_{n-1}+p_{n-4}$ and $q=\sum p_i$.
Then,
$$1\le\gcd(p_1,...,p_n)\le\gcd(p_{n-1},p_{n-4})=1$$
so by definition $C(x)=q$. Let $p=(p_i)_{i\in[n]}$. Also, by definition of $\beta_n$ and that $(x,y)\in\NE$ we get for every $j\in[n]$ that $(p^T\beta_n)_j=(p^T\beta_n)_n=2p_{n-1}+p_{n-4}$, which by Lemma \ref{lmm:det beta_n bn identity} equals $\frac{|\det\beta_n|}{g_n}$. So,
\begin{align*}
    \p{\frac{g_n}{|\det\beta_n|}p}^T\beta_n=\bf1
\end{align*}
Therefore, by Cramer's rule,
\begin{align*}
    \frac{|K(\beta_n)|}{|\det\beta_n|}=\sum_ip_i=\frac{qg_n}{|\det\beta_n|}\Rightarrow C(x)=q=\frac{|K(\beta_n)|}{g_n}
\end{align*}
\qed\end{proof}

\begin{proof}[Proof of Lemma \ref{lmm:gcd upper bound}]
For convenience, abuse notation and set $b_n$ to be $|b_n|$ and set $\rho$ to be $|\rho|$ (so that $b_n,\rho>0$ for 
$n\ge3$). Let $\rho_n=\frac{b_{n+1}}{b_n}\in\Q$ and $g_n=\gcd(b_n,b_{n+1})=\gcd(b_n,\rho_nb_n)$. Asymptotically, $b_n\approx\rho^n$, so consider the computation of $\gcd(\rho^n,\rho^{n+1})$ using the Euclidean algorithm. The first iterations yield
\begin{align*}
    \gcd(\rho^n,&\rho\cdot\rho^n)\\
    \gcd(\rho^n,&(\rho-1)\rho^n)\\
    \gcd((3-2\rho)\rho^n,&(\rho-1)\rho^n)\\
    \gcd((3-2\rho)\rho^n,&(13\rho-19)\rho^n)\\
    \vdots
\end{align*}
As $\rho\notin\Q$, the algorithm will not terminate. Therefore we get a pair of degree-1 polynomial sequences: $p_k=p_k^1\rho-p_k^0,q_k=q_k^1\rho-q_k^0\in\Z[\rho]$ such that the $k$th iteration of the Euclidean algorithm yields $\gcd(\rho^n,\rho\cdot\rho^n)=\gcd(p_k\rho^n,q_k\rho^n)$. Define $p_1=1,q_1=\rho$. Again, as $\rho\notin\Q$, we get $p_kq_k>0$ for all $k\ge1$. In particular,
\begin{align*}
    \forall k\in\N~~~0<p_{2k-1}=p_{2k}<q_{2k-1},~0<q_{2k}=q_{2k+1}<p_{2k}
\end{align*}
Let $M_k$ be $p_k$ if $k$ is odd, and $q_k$ if even.

\begin{claim}
$0<M_k<2^{-\frac{k-1}{2}}$.
\end{claim}

\begin{proof}
The inequality to 0 is clear. Let $k\ge3$ be odd. We show that $p_k<\frac{1}{2}p_{k-1}$.
We have $p_k<p_{k-1}$, and $p_k=p_{k-1}-uq_{k-1}$ for the largest integer $u>0$ such that $p_k>0$. Therefore $p_k\le\frac{1}{2}p_{k-1}$. Also, $\rho\notin\Q$ so $p_k<\frac{1}{2}p_{k-1}$.
Therefore, $M_k=p_k<\frac{1}{2}p_{k-2}=\frac{1}{2}M_{k-2}<...<2^{-\frac{k-1}{2}}M_1=2^{-\frac{k-1}{2}}$.
For an even $k$ we similarly get $q_k<\frac{1}{2}q_{k-1}$. Then, $M_k=q_k<\frac{1}{2}q_{k-2}M_{k-2}<...<2^{-\frac{k-1}{2}}M_2<2^{-\frac{k-1}{2}}$.
\qed\end{proof}

Now we consider the evaluations $M_k(\rho_n)$.
Lemma \ref{lmm:anbn linear independence} implies that $\frac{b_{l+1}}{b_l}\ne\frac{b_{m+1}}{b_m}$ for every $l,m\ge7$. Since the root of $M_k$ is rational, there is at most one $l\ge7$ such that $M_k(\rho_l)=0$, and for all other $l$s $M_k(\rho_l)>0$. So given $n$, let $K_M(7)=1$ and for $n>7$ let
\begin{align*}
    K_M(n)=\begin{cases}
        \min\s{k\mid M_k(\rho_n)=0}&\s{k\mid M_k(\rho_n)=0}\ne\emptyset\\
        \max_{n'<n}K_M(n')+1&else
    \end{cases}
\end{align*}
Then by definition, $M_k(\rho_n)>0$ for all $k<K_M(n)$. Also, $K_M$ is an injective $\s{7,...}\to\N$ function. We use the following claim:
\begin{claim}
Except on a subsequence $m_n=\omega(n)$, $K_M(n)=\Omega(n)$.
\end{claim}

\begin{proof}
Suppose for contradiction that there is a $c>0$ and a subsequence satisfying $m_n\le cn$ for all $n$ and $K_M(m_n)\le f(n)=o(n)$. $c\ge1$ because $m_n\ge n$. Consider the first $\frac{n}{c}$ terms of $K_M$. At least $n$ of them will be of the form $K_M(m_l)$, and they satisfy $\max_lK_M(m_l)\le f\p{\frac{n}{c}}$, so because $K_M$ is injective and positive, this means that there are at most $f\p{\frac{n}{c}}$ terms of the form $K_M(m_l)$ in the first $\frac{n}{c}$ terms of $K_M$. But $f\p{\frac{n}{c}}<n$ for sufficiently large $n$, so we get a contradiction. Therefore $m_n=\omega(n)$.
\qed\end{proof}

Let $m_n$ be the subsequence from the claim, and suppose $K_M(n)\ge cn$ for some $c>0$ and for $n\ne m_l$. When running the Euclidean algorithm on $b_n,b_{n+1}$ (but not optimally, only so as to get the intermediate gcds $\gcd(p_ib_n,q_ib_n)$), we get for all sufficiently large $n\ne m_l$ that
\begin{align*}
    g_n=\gcd(b_n,\rho_nb_n)=\gcd(p_{K_M(n)}b_n,q_{K_M(n)}b_n)\le M_{K_M(n)}(\rho_n)b_n<2\cdot2^{-\frac{cn-1}{2}}
\end{align*}
where the last inequality is due to $M_{K_M(n)}$'s continuity, and that $\rho_n\to\rho$.
The lemma now follows for $\gamma=\sqrt2+0.01$.
\qed\end{proof}

\begin{proof}[Proof of Corollary \ref{cor:x2 is balanced}]
Let $p_1,...,p_n$ be the numerators of $x$, chosen so that $p_{n-1}=|b_n|,p_{n-4}=|a_n|,p_n=2p_{n-1}+p_{n-4}$, and for all $k\ne0,1,4$, $p_{n-k}=a_kp_{n-1}+b_kp_{n-4}$. Suppose wlog that $b_n>0$. Then by Lemma \ref{lmm:anbn}
\begin{align*}
    p_{n-k}=a_kb_n-b_ka_n=(b_k+b_{k+1})b_n-b_k(b_n+b_{n+1})=b_nb_{k+1}-b_{n+1}b_k
\end{align*}
Consider Equation \ref{eq:bn}. $t_k=2\Re(w_2z^k)$, so that $b_k=\frac{1}{3}+w_1\rho^n+t_l$. Then similarly to the proof of Lemma \ref{lmm:beta_n nes}, we have
\begin{align*}
    &b_nb_{k+1}-b_{n+1}b_k\\
    &=t_nt_{k+1}-t_{n+1}t_k+\frac{w_1}{3}(\rho-1)(\rho^k-\rho^n)+\frac{1}{3}(t_n+t_{k+1}-t_{n+1}-t_k)+\\
    &~~~~w_1(\rho^nt_{k+1}+\rho^{k+1}t_n-\rho^{n+1}t_k-\rho^kt_{n+1})\\
    &=\frac{w_1}{3}(\rho-1)(\rho^k-\rho^n)+\frac{1}{3}(t_{k+1}-t_k)+w_1(\rho^nt_{k+1}+\rho^{k+1}t_n-\rho^{n+1}t_k-\rho^kt_{n+1})+o_n(1)
\end{align*}
We further have
\begin{align*}
    \frac{b_nb_{k+1}-b_{n+1}b_k}{b_n}&=\frac{\frac{w_1}{3}(\rho-1)(\rho^k-\rho^n)+w_1(\rho^nt_{k+1}-\rho^{n+1}t_k)}{b_n}+o_n(1)\\
    &=\frac{1}{3}(\rho-1)(\rho^{k-n}-1)+t_{k+1}-\rho t_k+o_n(1)
\end{align*}
So for every $1<k<n$:
\begin{align*}
    \frac{x_{n-k}}{x_1}&=\frac{p_{n-k}}{p_1}=\frac{\frac{p_{n-k}}{b_n}}{\frac{p_1}{b_n}}=\frac{\frac{1}{3}(\rho-1)(\rho^{k-n}-1)+t_{k+1}-\rho t_k+o_n(1)}{\frac{1}{3}(\rho-1)(\rho^{1-n}-1)+t_2-\rho t_1+o_n(1)}\\
    &\le\frac{c(\rho-1)(\rho^{k-n}-1)}{(\rho-1)(\rho^{1-n}-1)+t_2-\rho t_1}+o_n(1)\\
    &\le c\p{1-\frac{t_2-\rho t_1}{(\rho-1)(\rho^{1-n}-1)+t_2-\rho t_1}}+o_n(1)\\
    &=O(1)
\end{align*}
For some $c=\Theta(1)$, and this is because $t_{k+1}-\rho t_k=O(|z|^k)=o_k(1)$ and $\frac{1}{3}(\rho-1)(\rho^{k-n}-1)=\Theta(1)$. For $x_n$ we have by definition, for some $c'=O(1)$:
\begin{align*}
    x_n=x_1+x_3\le c'x_1\Rightarrow\frac{x_n}{x_1}\le c'=O(1)
\end{align*}
Hence $\max_{i,j}\frac{x_i}{x_j}=O(1)$.
Therefore, all the numerators are at least $\frac{2^{\Theta(n)}}{\Theta(n)}$, so they would require in total $\Theta(n^2)$ bits.
\qed\end{proof}

\section{Omitted Proofs in Section \ref{subsec:constant}}\label{appendix:constant}
\begin{lemma}\label{lmm:pi tau matrix property}
Let $y$ and let $x=\pi(y)$. Then $B^Tx=\tau(By)$.
\end{lemma}
\begin{proof}
Let $i\in[n]$.
\begin{align*}
    (B^Tx)_i=\sum_jB_{j,i}y_{\pi(j)}=\sum_jB_{\tau(i),\pi(j)}y_{\pi(j)}=(By)_{\tau(j)}
\end{align*}
so $B^Tx=\tau(By)$.
\qed\end{proof}

A property of zero-sum games mentioned in \citet{avis2010enumeration} will also be used. For completeness, we state it here in terms of constant-sum games and prove it.
\begin{lemma}\label{lmm:2p0s ne set is a rect}
Let $(A,B)$ be a constant-sum game with $A+B={\bf1}_{n\times n}$, and let $\pi_i(\NE(A,B))$ be the equilibrium strategies of player $i$. Then $\NE(A,B)=\pi_1(\NE(A,B))\times\pi_2(\NE(A,B))$.
\end{lemma}
\begin{proof}
By the minimax theorem \cite{v1928theorie} $(x,y)\in\NE$ iff both $x\in\argmax_s\min_ts^TAt$ and $y\in\argmax_t\min_ss^TBt$. The lemma now follows.
\qed\end{proof}

\begin{proof}[Proof of Lemma \ref{lmm:2p0s unique fully mixed ne lemma}]
First, we show that $(\Tilde{x},\pi^{-1}(\Tilde{x}))$ is a NE. We have $B^T\Tilde{x}={\bf1}_nv$ for some $v>0$. Since $(\Tilde{x},U_n)$ is fully mixed and $\det B\ne0$, we get that $B$ contains no all-1 columns. Therefore, $v=(B^T\Tilde{x})_i<\inp{{\bf1},\Tilde{x}}=1$. For $\pi^{-1}(\Tilde{x})$, we have
\begin{align*}
    ({\bf1}_{n\times n}-B)\pi^{-1}(\Tilde{x})={\bf1}_{n\times n}\pi^{-1}(\Tilde{x})-B\pi^{-1}(\Tilde{x})={\bf1}_n-\tau^{-1}(B^T\Tilde{x})=(1-v){\bf1}_n
\end{align*}
(by Lemma \ref{lmm:pi tau matrix property}). Therefore by the best-response condition \cite{roughgarden2010algorithmic}, it is a NE of the constant-sum game.
To see that it is unique, by the well-known matrix determinant lemma and the given,
\begin{align*}
    \det({\bf1}_{n\times n}-B)=\p{1-{\bf1}^TB^{-1}{\bf1}}\det(-B)=(-1)^n\p{1-\frac{K(B)}{\det B}}\det B\ne0
\end{align*}
Now let $(x,y)\in\NE({\bf1}_{n\times n}-B,B)$ with supports $I,J$. By Lemma \ref{lmm:2p0s ne set is a rect}, $(x,\pi^{-1}(\Tilde{x}))\in\NE$. $x$ satisfies
\begin{align*}
    B^Tx=v'{\bf1}_n,~~~\inp{x,{\bf1}_n}=1
\end{align*}
Since $B^T$ is invertible, $x=(B^T)^{-1}v'{\bf1}_n$, and $v'$ is such that $\inp{x,{\bf1}_n}=1$. By invertibility and uniqueness, we get $x=\Tilde{x}$. Similarly, we get $y=\pi^{-1}(\Tilde{x})$. Hence $(\Tilde{x},\pi^{-1}(\Tilde{x}))$ is the unique NE of the constant-sum game.
\qed\end{proof}

\begin{proof}[Proof of Theorem \ref{thm:2.3.2 formal}]
Let $p_1,...,p_n$ be the first $n$ primes, and let $N,B$ be the integer and the matrix from the proof of Theorem \ref{thm:2^sqrt n lower bound}.
$\sum_jB_{i,j}\le p_n=\Tilde\Theta(\sqrt{N})$. Therefore, by Lemma \ref{lmm:K det inequality strengthened}, $|K(B)|\ge\Tilde\Theta(\sqrt{N})|\det B|$. Also, $\det B=(-1)^{N-1}\prod_kp_k\ne0$ as was shown in Theorem \ref{thm:2^sqrt n lower bound}. Again, by Theorem \ref{thm:2^sqrt n lower bound}, $(I_N,B)$ has a unique NE and it is fully mixed. Also, by the construction of $B$, it is $(\pi,\pi)$-symmetric for the cycle $\pi=(n~n-1~...~1)$. Now we can use Lemma \ref{lmm:2p0s unique fully mixed ne lemma} to conclude that $({\bf1}_{N\times N}-B,B)$ has a unique NE $(\Tilde{x},\pi^{-1}(\Tilde{x})$ and it is fully mixed. Furthermore, by Theorem \ref{thm:2^sqrt n lower bound}, there holds
\begin{align*}
    C_1({\bf1}_{N\times N}-B,B)&=C_2({\bf1}_{N\times N}-B,B)=C(\Tilde{x})=C(\pi^{-1}(\Tilde{x}))\\
    &=\p{\prod_kp_k}\p{n+1+\sum_{k=1}^n\frac{1}{p_k}}=2^{\Tilde{\Theta}(\sqrt N)}
\end{align*}
Then, a similar argument to Corollary \ref{cor:sqrt n lower bound} establishes this lower bound for all positive integers.
\qed\end{proof}

\section{Special Cases}\label{sec:special}
\subsection{$2\times2$ Games}
\label{sec:2x2}
Let $A,B\in\Z^{2\times2}$. If $\PNE(A,B)\ne\emptyset$, then $C_i(A,B)=1$. Suppose that $\PNE(A,B)=\emptyset$.

Wlog $A_{11}>A_{21}$. This implies that $A_{12}<A_{22}$, and $B_{11}<B_{12}$ and $B_{21}>B_{22}$. By definition, the only equilibria of the game are given by $p,q\in(0,1)$ that satisfy
\begin{align*}
    qA_{11}+(1-q)A_{12}=qA_{21}+(1-q)A_{22},\\
    pB_{11}+(1-p)B_{21}=pB_{12}+(1-p)B_{22}
\end{align*}
Solving this yields
\begin{align*}
    p=\frac{B_{21}-B_{22}}{B_{12}-B_{11}+B_{21}-B_{22}},~q=\frac{A_{22}-A_{12}}{A_{11}-A_{21}+A_{22}-A_{12}}
\end{align*}
after reducing these 2 fractions, using the equality $\gcd(a,a+b)=\gcd(a,b)$, we get

\begin{theorem}
If $\PNE(A,B)=\emptyset$, then
\begin{align*}
    C_1(A,B)=\frac{B_{12}-B_{11}+B_{21}-B_{22}}{\gcd(B_{12}-B_{11},B_{21}-B_{22})},~C_2(A,B)=\frac{A_{11}-A_{21}+A_{22}-A_{12}}{\gcd(A_{11}-A_{21},A_{22}-A_{12})}
\end{align*}
\end{theorem}

\subsection{Permutation Games}
\label{sec:perm}
Let $\sym_n$ be the symmetric group of size $n$. Given $\pi\in\sym_n$, let $P_\pi$ be its permutation matrix, so that $P_\pi e_i=e_{\pi(i)}$. Then $P_\pi^T=P_{\pi}$. Write $\pi=\pi_1\cdots\pi_a$ and $\tau=\tau_1\cdots\tau_b$ the permutations in cycle form. Also write $\pi^{-1}\tau=(\pi^{-1}\tau)_1\cdots(\pi^{-1}\tau)_d$. Let $|(\pi^{-1}\tau)_k|$ be the length of the $k$th cycle, and let $J_k$ be the set of indices in that cycle (e.g. $\s{1,3,4}$ for the first cycle of $(143)(25)$). Let $\pi,\tau\in\sym_n$, and let $A=P_\pi,B=P_\tau\in\Z_2^{n\times n}$. Consider the game $(A,B)$.
We prove the following:

\begin{theorem}
$\NE(P_\pi,P_\tau)=\s{(U_{\pi(J)},U_J)\mid D\subseteq[d],~J=\bigcup_{k\in D}U_{J_k}}$, and therefore
$$C_1(A,B)=C_2(A,B)=\min_{k\le d}|(\pi^{-1}\tau)|_k$$
\end{theorem}
\begin{proof}
Take $(x,y)\in\NE(P_\pi,P_\tau)$ with supports $I,J$. Then $P_\tau^Tx=u{\bf1}_J+s^1,P_\pi y=v{\bf1}_I+s^2$ for some $u,v>0$, where $0\le s^1\le{\bf1}_{[n]\setminus J}\cdot u$ and $0\le s^2\le{\bf1}_{[n]\setminus I}\cdot v$. Looking more closely,
\begin{align*}
    P_\tau^Tx=P_{\tau^{-1}}x=\sum_{i\in I}x_ie_{\tau^{-1}(i)}=u{\bf1}_J+s^1
\end{align*}
So because $u>0$ we get $J\subseteq\tau^{-1}(I)$. Similarly, $I\subseteq\pi(J)$. So $\tau(J)\subseteq I\subseteq\pi(J)$, and these are all equalities because $\pi,\tau$ are permutations. In particular $s^1=s^2={\bf0}$, so $x=U_I,y=U_J$. Also, $J=\pi^{-1}\tau(J)$, so $J$ is a union of cycles of $\pi^{-1}\tau$, so $J=\bigcup_{k\in D}U_{J_k}$ for some $D\subseteq[d]$. Then, $I=\pi(J)$.
On the other hand, given any such $J$ and $I=\pi(J)$, the calculations above verify that $(U_{\pi(J)},U_J)\in\NE(P_\pi,P_\tau)$.
\qed\end{proof}

As a corollary, we also obtain $\NE(P_\pi,P_\tau)=\NE(I_n,P_{\pi^{-1}\tau})$.

\end{document}